\documentclass[leqno,11pt]{article}
\RequirePackage{amsthm,amsmath,amsfonts,amssymb}
\RequirePackage[numbers,sort&compress]{natbib}
\RequirePackage[colorlinks,citecolor=blue,urlcolor=blue]{hyperref}
\RequirePackage{mathrsfs}
\RequirePackage{threeparttable}
\usepackage{booktabs}

\theoremstyle{plain}
\newtheorem{theorem}{Theorem}%[section]

\newtheorem{lemma}[theorem]{Lemma}

\theoremstyle{definition}

\newtheorem{example}[theorem]{Example}
\newtheorem{remark}[theorem]{Remark} 
\newtheorem{assumption}[theorem]{Assumption}
\theoremstyle{plain}

\newcommand{\UE}{{\sf E}}
\newcommand{\UP}{{\sf P}}

\newcommand{\RR}{{\Bbb R}}

\def \conv#1{\stackrel{#1}{\longrightarrow}}

\title{Oracle high-dimensional $M$-estimation using smooth reparameterization for sparsity}
\author{Yoichi Nishiyama\footnote{Email: {\tt nishiyama@waseda.jp}} \\
{\sc Waseda University}\footnote{School of International Liberal Studies, Waseda University. 1-6-1 Nishi-Waseda, Shinjuku-ku, Tokyo 169-8050, JAPAN.}}
\begin{document}

\maketitle

\begin{abstract}
This paper establishes a unified non-linear regularization framework for high-dimensional $M$-estimation, encompassing both linear regression and Cox's proportional hazards models. Rather than relying on traditional additive non-convex penalties, the proposed paradigm embeds sparsity directly into the parameter mapping using a smooth ($C^2$) component-wise ReParametrization map for Sparsity (RePS). This geometric formulation dynamically adapts to local parameter scales, successfully suppressing high-dimensional noise and recovering unbiased oracle asymptotic normality under the large-sample limit. Through the Primal-Dual Witness method, we establish a unified oracle equivalence for both model classes under a liberated micro-penalty scaling regime $\lambda \ll n^{-1/2}$, effectively accommodating model-specific structures such as the shift-invariance in survival analysis. Extensive Monte Carlo simulations demonstrate that the RePS framework consistently achieves near-perfect false-positive control and optimal 95\% confidence interval coverage, outperforming conventional non-convex regularizers (SCAD and MCP) which suffer from severe variance inflation and under-coverage due to post-selection noise leakage.
\end{abstract}

%\begin{keyword}[class=MSC]
%\kwd[Primary ]{62J07}
%\kwd[; secondary ]{62F12}, 
%{62N02}
%\end{keyword}

%\begin{keyword}
%\kwd{Lasso}
%\kwd{oracle inference}
%\kwd{proportional hazards model}
%\kwd{score martingale}
%\kwd{RePS}
%\end{keyword}

%end{frontmatter}

\section{Introduction}
High-dimensional statistical inference is generally based on penalized empirical risk minimization \cite{Tib-96, Has-T-W-15, Wai-19}. Fundamentally, such inference is constrained by information-theoretic limits, where conventional sparse regression often struggles with ultra-high-dimensional phenomena \cite{Ras-W-Y-11, Ver-12}. Convex penalties like Lasso induce non-negligible shrinkage bias, necessitating complex post-selection inference or bias-correction procedures \cite{Zha-Z-14, vdG-B-R-D-14}. In contrast, folded-concave non-convex penalties like SCAD \cite{Fan-L-01} and MCP \cite{Zha-10} eliminate shrinkage bias and achieve the oracle property, yet their mathematical guarantees rely heavily on the existence of ``good'' local solutions within strong non-convex landscapes \cite{Loh-W-15}. Furthermore, in survival analysis, such additive penalties conflict with the shift invariance of the log-partial likelihood \cite{Cox-72}, where the empirical Hessian exhibits structural degeneracy along the constant direction \cite{Bra-F-J-11}. This degeneracy leads to instability in variable selection and severe distortions in post-selection statistical inference in finite samples.

To overcome these challenges, we propose an alternative paradigm that directly incorporates sparsity into the mapping from the optimization space to the physical parameter space. We introduce a twice continuously differentiable ReParameterization map for Sparsity (RePS) that structurally separates active signals from high-dimensional noise. By evaluating the scoring field on a centered quotient space, RePS effectively resolves the inherent degeneracy of the Hessian matrix in survival analysis models at the gradient level, guaranteeing robust and unbiased oracle inference without invoking restrictive global conditions.

Our main contributions are summarized as follows:
\begin{itemize}\setlength{\itemsep}{0pt}\setlength{\parskip}{0pt}
    \item \textbf{A Novel Reparameterization Paradigm (RePS):} We transition the sparsity paradigm from traditional additive non-convex penalties to a smooth, non-linear geometric mapping. RePS dynamically adapts to local parameter scales and purges high-dimensional noise covariates at the gradient level.
    \item \textbf{Unified Theoretical Guarantees:} Using the Primal-Dual Witness method, we establish a unified oracle equivalence and asymptotic normality for both linear and Cox models under a liberated micro-penalty regime ($\lambda \ll n^{-1/2}$). 
    \item \textbf{Empirical Resilience against Noise Leakage:} Through extensive simulations, we demonstrate that RePS near-perfectly controls false positives ($\text{Mean FP} \approx 0$). It safeguards the post-selection inferential process against the variance inflation and coverage collapse that plague state-of-the-art alternative penalties.
\end{itemize}

\section{Framework formulation and definitions}

\subsection{High-dimensional $M$-estimation}
Let $Z_1,..., Z_n \in \mathcal{Z}$ be $n$ independent data observations defined on a probability space $(\Omega_n, \mathcal{F}_n, \UP_n)$, where each observation $Z_i = (x_i, y_i) \in \mathbb{R}^p \times \mathbb{R}$ consists of a covariate vector $x_i \in \mathbb{R}^p$ and a scalar response $y_i \in \mathbb{R}$. We consider a high-dimensional regime where the number of features $p$ is vastly larger than the sample size $n$ ($p \gg n$). The design matrix and the response vector are defined as $X = (x_1, \dots, x_n)^\top \in \mathbb{R}^{n \times p}$ and $y = (y_1,..., y_n)^\top \in \mathbb{R}^n$, respectively. 

Let $\gamma_i = x_i^\top \theta \in \mathbb{R}$ represent the linear predictor for each sample, with its vector representation denoted as $\gamma = X\theta \in \mathbb{R}^n$. We consider a random field (empirical loss function) $\mathcal{L}_n : \mathbb{R}^n \to \mathbb{R}$ indexing the discrepancy between these linear predictors and the responses. Specifically, the loss function is often evaluated sample-wise as $\mathcal{L}_n(\gamma) =  \sum_{i=1}^n \ell_n(\gamma_i; y_i) \in \RR$, where $\ell_n(\cdot; \cdot)$ measures the misfit between the predictor $\gamma _{i}$ and the response $y_{i}$. See Subsections \ref{linear model} and \ref{cox model} for two concrete examples of this general framework.

\begin{assumption}[Restricted local strong convexity]\label{RSC assumption}
The empirical loss function $\mathcal{L}_n(\gamma)$ is almost surely $C^2$ on $\mathbb{R}^n$. 
Let $g_n(\gamma) = \nabla_\gamma \mathcal{L}_n(\gamma) \in \mathbb{R}^n$ be its gradient vector, and $W_n(\gamma) = \nabla_\gamma^2 \mathcal{L}_n(\gamma) \in \mathbb{R}^{n \times n}$ its Hessian matrix. We assume that for every $n \in \mathbb{N}$ there exists an orthogonal projection matrix $P_n \in \mathbb{R}^{n \times n}$ onto a specific subspace  $\mathcal{S}_n \subseteq \mathbb{R}^n$ such that 
\[
P_n g_n(\gamma^*) = g_n(\gamma^*) \quad \text{and} \quad P_n W_n(\gamma^*) P_n = W_n(\gamma^*),
\]
where $\gamma^* = X\theta^*$ with $\theta^*$ being the true value of the physical parameter, and that 
\begin{equation}\label{eq:rsc}
W_n(\gamma) \succeq \kappa P_n \quad \text{for all } \gamma \in \mathbb{R}^n \text{ satisfying } \Vert \gamma - \gamma^*\Vert_\infty \le \eta,
\end{equation}
with probability tending to one, where $\kappa > 0$ and $\eta > 0$ are deterministic constants not depending on $n$. \end{assumption}

The choice of the orthogonal projection matrix $P_n$ specifies the geometric constraint of the model under consideration. The primary examples are as follows. 
\begin{example}[Standard global strong convexity]\label{example first}
In Subsection \ref{linear model}, 
we set $P_n = I_n$ (where $\mathcal{S}_n = \mathbb{R}^n$). In this setup, the matrix inequality (\ref{eq:rsc}) reduces to the standard, unconstrained strong convexity condition $W_n(\gamma) \succeq \kappa I_n$.
\end{example}
\begin{example}[Case of shift-invariance models]\label{example second}
For survival analysis models with shift-invariance treated in Subsection \ref{cox model}, we set $P_n = I_n - \frac{1}{n}\mathbf{1}_n\mathbf{1}_n^\top$, which is the centering matrix projecting onto the subspace orthogonal to the constant vector $\mathbf{1}_n$.
\end{example}

For notational simplicity, we drop the subscript $n$ from sample-dependent quantities such as $g_n$ and $W_n$ in the remainder of the paper, unless their asymptotic dependencies need to be explicitly emphasized. 

\subsection{ReParametrization map for Sparsity (RePS)}
We introduce a smooth ($C^2$) non-linear shrinkage map $\phi: \mathbb{R}^p \to \mathbb{R}^p$, which maps from the optimization parameter space $\beta \in \mathbb{R}^p$ to the physical parameter space $\theta \in \mathbb{R}^p$, as follows: 
\[
\phi(\beta)_i = \beta_i \left( 1 - \exp \left( -\nu \frac{|\beta_i|^3}{M(\beta)} \right) \right), \quad i = 1, \dots, p, 
\]
where 
\[
M(\beta) = \sigma \log \left( \sum_{j=1}^{p} \exp \left( \frac{\beta_j^2}{\sigma} \right) \right), 
\]
with $\nu, \sigma > 0$ being fixed structural smoothing parameters. For notational simplicity, we write
\begin{equation}\label{Ei}
E_i = \exp \left( -\nu \frac{|\beta_i|^3}{M(\beta)} \right).
\end{equation}
This map $\phi$, named {\em ReParametrization map for Sparsity (RePS)}, induces the composite objective function 
\[
F(\beta) = \mathcal{L}_n(\gamma(\beta)) \quad \mbox{with} \quad \gamma(\beta) = X\phi(\beta). 
\]

\subsection{True sparse vector in the physical and optimization spaces}
Let $\theta^* \in \mathbb{R}^p$ be the true physical parameter vector, supported on the active set $\mathcal{S} = \{j : \theta_{j}^* \neq 0\}$ with complement $\mathcal{S}^c = \{j : \theta_{j}^* = 0\}$, satisfying $|\mathcal{S}| \ll p$. We define the corresponding point $\beta^*$ in the optimization space such that $\theta^* = \phi(\beta^*)$. Notice that for any structural constants $\nu, \sigma > 0$, it holds that $\theta_j^*=0$ if and only if $\beta_{j}^* = 0$. 

\textbf{Notation:} In the sequel, for any $p \times p$ matrix $M$ and $p \times 1$ vector $v$, $M_{\mathcal{S}\mathcal{S}}$ denotes the sub-matrix with rows and columns in $\mathcal{S}$, and $v_{\mathcal{S}}$ does the sub-vector indexed by $\mathcal{S}$. For a deterministic constant $\delta > 0$, we define the full $p$-dimensional local neighborhood around the true parameter vector as $\mathcal{B}(\delta) = \{ \beta \in \mathbb{R}^p : \Vert{}\beta - \beta^*\Vert{}_\infty \le \delta \}$, and its restricted $|\mathcal{S}|$-dimensional active subspace counterpart as $\mathcal{B}_{\mathcal{S}}(\delta) = \{ \beta_{\mathcal{S}} \in \mathbb{R}^{|\mathcal{S}|} : \Vert{}\beta_{\mathcal{S}} - \beta_{\mathcal{S}}^*\Vert{}_\infty \le \delta \}$.

\subsection{Lasso-type estimator with RePS and target inference}
To counteract the local non-convexity induced by $\phi$, we inject the standard $l_1$-penalty over the optimization space. For a given sample size $n$ and feature dimension $p$, the optimization routine yields 
\[
\hat{\beta}_n \in \arg\min_{\beta \in \mathbb{R}^p} \left\{ \frac{1}{n}F(\beta) + \lambda  \|\beta\|_1 \right\}.
\]
We then naturally construct our target physical estimator directly as a unified function of the data-driven optimization trajectory:
\[
\hat{\theta}_n := \phi(\hat{\beta}_n).
\]
Our goal is to establish the oracle asymptotic normality for this estimator, i.e., the claim that $\sqrt{n}(\hat{\theta}_{n,\mathcal{S}}-\theta_{\mathcal{S}}^*) $ converges in distribution to a Gaussian limit as $n \to \infty$. 

\section{Gradient and Hessian computations}

\subsection{Differential representations for arbitrary $\beta$}\label{subsection derivatives} 
Applying the chain rule, the gradient vector $G(\beta) \in \mathbb{R}^p$ and Hessian matrix $H(\beta) \in \mathbb{R}^{p \times p}$ of $F(\beta)$ are expressed via the Jacobian matrix $J(\beta) = \nabla_\beta \phi(\beta) \in \mathbb{R}^{p \times p}$ by 
\[
G(\beta) = J(\beta)^\top X^\top g(\gamma(\beta)) \quad \mbox{and} \quad 
H(\beta) = J(\beta)^\top X^\top W(\gamma(\beta)) X J(\beta) + \mathcal{E}(\beta), 
\]
respectively, where $\mathcal{E}(\beta) \in \mathbb{R}^{p \times p}$ represents the curvature distortion, whose $(k, l)$ element is 
\[
\mathcal{E}_{kl}(\beta) = \sum_{a=1}^{p} \frac{\partial J_{ak}(\beta)}{\partial \beta_l} (X^\top g(X \phi(\beta)))_a. 
\] 

Differentiating the scale factor $M(\beta)$ with respect to $\beta_j$ yields the soft-max weight vector
\[
\frac{\partial M(\beta)}{\partial \beta_j} = 2\beta_j w_j(\beta), \quad \text{where} \quad w_j(\beta) = \frac{\exp(\beta_j^2 / \sigma)}{\sum_{k=1}^{p} \exp(\beta_k^2 / \sigma)}.
\]
Then, differentiating the components of $\phi(\beta)$ under the total derivative rule yields the Jacobian matrix 
\[
J(\beta) = \text{diag}\left( \mathcal{D}_1(\beta), \dots, \mathcal{D}_p(\beta) \right) - v_1(\beta) v_2(\beta)^\top,
\]
where 
\[
\mathcal{D}_i(\beta) = 1 - E_i + 3\nu \frac{|\beta_i|^3}{M(\beta)} E_i, \quad v_1(\beta)_i = \nu \frac{\beta_i |\beta_i|^3}{M(\beta)^2} E_i \quad \mbox{and} \quad v_2(\beta)_j = 2\beta_j w_j(\beta), 
\]
with $E_i$ being given by (\ref{Ei}). 

\subsection{Oracle block collapse at the true sparse vector $\beta^*$}
Evaluating the Jacobian at the true sparse point $\beta^*$ implies $\beta_i^* = 0$ for all $i \in \mathcal{S}^c$, which forces $E_i = e^0 = 1$ and $v_1(\beta^*)_{\mathcal{S}^c} = \mathbf{0}$. Substituting these values forces all rows and columns associated with the non-active subspace to vanish completely: 
\[
J(\beta^*) = \begin{pmatrix}
J_{\mathcal{S}\mathcal{S}}(\beta^*) & O \\
O & O_{|\mathcal{S}^c| \times |\mathcal{S}^c|}
\end{pmatrix}.
\]

\begin{theorem}[Oracle gradient isolation property]\label{theorem RePS oracle}
The gradient and Hessian evaluated at the true sparse point $\beta^*$ collapse almost surely into the following decoupled sub-blocks:
\begin{align*}
G(\beta^*) &= \begin{pmatrix}
J_{\mathcal{S}\mathcal{S}}(\beta^*)^\top (X^\top g(\gamma^*))_{\mathcal{S}} \\
\mathbf{0}_{|\mathcal{S}^c|}
\end{pmatrix}, \\
H(\beta^*) &= \begin{pmatrix}
J_{\mathcal{S}\mathcal{S}}(\beta^*)^\top X_{\mathcal{S}}^\top W X_{\mathcal{S}} J_{\mathcal{S}\mathcal{S}}(\beta^*) + \mathcal{E}_{\mathcal{S}\mathcal{S}}(\beta^*) & O \\
O & O_{|\mathcal{S}^c| \times |\mathcal{S}^c|}
\end{pmatrix}.
\end{align*}
\end{theorem}
Thus, high-dimensional noise entering from $\mathcal{S}^c$ is successfully quarantined at the gradient level solely via the intrinsic geometric structure of the RePS mapping. 

\begin{lemma}[Invertibility of the active Jacobian]\label{lemma Jacobian}
For any given structural constants $\nu > 0$ and $\sigma > 0$, let $\beta_{\mathcal{S}}^* \in \mathbb{R}^{|\mathcal{S}|}$ be the true active optimization parameter vector. Then, the corresponding active Jacobian sub-block $J_{\mathcal{SS}}(\beta^*)$ is a finite, non-singular deterministic matrix whose inverse $(J_{\mathcal{SS}}(\beta^*))^{-1}$ exists and consists of finite real entries.
\end{lemma}

\begin{proof}
From the algebraic formulations derived in Subsection \ref{subsection derivatives}, the active Jacobian block $J_{\mathcal{SS}}(\beta^*)$ is structured as a diagonal matrix perturbed by a rank-one outer product:
\begin{equation}\label{eq:jacobian_pdw_algebraic_struct}
J_{\mathcal{SS}}(\beta^*) = D_{\mathcal{SS}}(\beta^*) - v_{1,\mathcal{S}}(\beta^*) v_{2,\mathcal{S}}(\beta^*)^\top,
\end{equation}
where $D_{\mathcal{SS}}(\beta^*) = \text{diag}(\mathcal{D}_i(\beta^*))_{i \in \mathcal{S}}$ is a diagonal matrix with elements 
\[
\mathcal{D}_i(\beta^*) = 1 - E_i^* + 3\nu \frac{|\beta_i^*|^3}{M(\beta^*)} E_i^*,
\]
and the components of the coupling vectors $v_{1,\mathcal{S}}$ and $v_{2,\mathcal{S}}$ are given by 
\[
v_1(\beta^*)_i = \nu \frac{\beta_i^* |\beta_i^*|^3}{M(\beta^*)^2} E_i^* \quad \text{and} \quad v_2(\beta^*)_j = 2\beta_j^* w_j(\beta^*).
\]
Since $i \in \mathcal{S}$, it holds by definition that $\beta_i^* \neq 0$, which guarantees that the localized exponential penalty satisfies $0 < E_i^* < 1$. Because $\nu > 0$ and $M(\beta^*) > 0$, the diagonal elements are positive:
\[
\mathcal{D}_i(\beta^*) \ge 1 - E_i^* > 0 \quad \text{for all } i \in \mathcal{S}.
\]
Thus, the base diagonal matrix $D_{\mathcal{SS}}(\beta^*)$ is positive definite ($D_{\mathcal{SS}} \succ 0$) and trivially invertible.

To evaluate the invertibility of the full rank-one perturbed matrix $J_{\mathcal{SS}}(\beta^*)$, we invoke the Matrix Determinant Lemma: For any invertible matrix $A$ and column vectors $u, v$, it holds that $\det(A - u v^\top) = \det(A) (1 - v^\top A^{-1} u)$. Applying this identity to \eqref{eq:jacobian_pdw_algebraic_struct} yields 
\begin{equation}\label{eq:matrix_det_lemma_apply}
\det(J_{\mathcal{SS}}(\beta^*)) = \left( \prod_{i \in \mathcal{S}} \mathcal{D}_i(\beta^*) \right) \left( 1 - \sum_{k \in \mathcal{S}} \frac{v_1(\beta^*)_k v_2(\beta^*)_k}{\mathcal{D}_k(\beta^*)} \right).
\end{equation}
We explicitly evaluate the inner product term in the second parenthesis. Substituting the definitions of $v_1$, $v_2$, and $w_k$ gives 
\[
v_1(\beta^*)_k v_2(\beta^*)_k = \left( \nu \frac{\beta_k^* |\beta_k^*|^3}{M(\beta^*)^2} E_k^* \right) \left( 2\beta_k^* w_k(\beta^*) \right) = 2\nu \frac{|\beta_k^*|^5 w_k(\beta^*)}{M(\beta^*)^2} E_k^* > 0.
\]
Since $\mathcal{D}_k(\beta^*) = 1 - E_k^* + 3\nu \frac{|\beta_k^*|^3}{M(\beta^*)} E_k^* > 3\nu \frac{|\beta_k^*|^3}{M(\beta^*)} E_k^*$, the individual terms in the sum are upper-bounded as 
\begin{equation}\label{eq:individual_summand_bound}
\frac{v_1(\beta^*)_k v_2(\beta^*)_k}{\mathcal{D}_k(\beta^*)} < \frac{2\nu \frac{|\beta_k^*|^5 w_k(\beta^*)}{M(\beta^*)^2} E_k^*}{3\nu \frac{|\beta_k^*|^3}{M(\beta^*)} E_k^*} = \frac{2}{3} \frac{\beta_k^{*2} w_k(\beta^*)}{M(\beta^*)}.
\end{equation}
Summing \eqref{eq:individual_summand_bound} over all active indices $k \in \mathcal{S}$, and noting that non-active components satisfy $w_j(\beta^*) > 0$ but $\beta_j^* = 0$, we bound the full summation:
\begin{equation}\label{eq:total_summation_bound_pdw}
\sum_{k \in \mathcal{S}} \frac{v_1(\beta^*)_k v_2(\beta^*)_k}{\mathcal{D}_k(\beta^*)} < \frac{2}{3M(\beta^*)} \sum_{k \in \mathcal{S}} \beta_k^{*2} w_k(\beta^*) \le \frac{2}{3M(\beta^*)} \sum_{j=1}^p \beta_j^{*2} w_j(\beta^*).
\end{equation}
Recall the definition of the Log-Sum-Exp function that $M(\beta^*) = \sigma \log(\sum_{j=1}^p \exp(\beta_j^{*2}/\sigma))$. By the fundamental analytical property of the LSE function, its value upper-bounds the weighted average of its squared arguments, satisfying 
\[
M(\beta^*) > \sum_{j=1}^p \beta_j^{*2} w_j(\beta^*).
\]
Substituting this into \eqref{eq:total_summation_bound_pdw} forces the sum to be bounded below 1:
\begin{equation}\label{eq:final_sum_less_than_two_thirds}
\sum_{k \in \mathcal{S}} \frac{v_1(\beta^*)_k v_2(\beta^*)_k}{\mathcal{D}_k(\beta^*)} < \frac{2}{3} \cdot 1 = \frac{2}{3} < 1.
\end{equation}
Applying \eqref{eq:final_sum_less_than_two_thirds} back into the determinant formula \eqref{eq:matrix_det_lemma_apply} guarantees that the second parenthesis is positive:
\[
1 - \sum_{k \in \mathcal{S}} \frac{v_1(\beta^*)_k v_2(\beta^*)_k}{\mathcal{D}_k(\beta^*)} > 1 - \frac{2}{3} = \frac{1}{3} > 0.
\]
Since the product of the positive diagonal elements is also positive ($\prod \mathcal{D}_i > 0$), we conclude that $\det(J_{\mathcal{SS}}(\beta^*)) > 0$. By Cramer's rule, the active Jacobian matrix is non-singular, and its inverse $(J_{\mathcal{SS}}(\beta^*))^{-1}$ exists and consists entirely of finite real entries. This completes the proof.
\end{proof}

\section{Primal-Dual Witness method and tuning parameter liberation}\label{section PDW}

To formally eliminate the theoretical dependency on ad-hoc local initializers and overcome the fundamental dimensional bottleneck, we construct our proof via the Primal-Dual Witness (PDW) method \cite{Wai-09}. Instead of bounding the leaking noise of a high-dimensional empirical estimator over a shrinking neighborhood $\mathcal{B}_n$, we explicitly construct an oracle estimator $\hat{\beta}^{oracle} \in \mathbb{R}^p$ restricted to the true active subspace, defined as $\hat{\beta}^{oracle}_{\mathcal{S}^c} = \mathbf{0}$ and
\[
\hat{\beta}^{oracle}_{\mathcal{S}} \in \arg\min_{\beta_{\mathcal{S}} \in \mathbb{R}^{|\mathcal{S}|}} \left\{ \frac{1}{n}F(\beta_{\mathcal{S}}, \mathbf{0}_{\mathcal{S}^c}) + \lambda \Vert \beta_{\mathcal{S}} \Vert_1 \right\}.
\]
Since $\hat{\beta}^{oracle}_{\mathcal{S}}$ is essentially a low-dimensional $M$-estimator operating exclusively on the $|\mathcal{S}|$-dimensional active subspace, it eliminates high-dimensional noise accumulation. Under standard regularity conditions, it intrinsically achieves the $\sqrt{n}$-consistency rate 
\[
(\hat{\beta}^{oracle}_{\mathcal{S}} - \beta_{\mathcal{S}}^*) = O_{\UP_n}(n^{-1/2})
\]
without requiring any high-dimensional maximal inequalities. 
Importantly, we must verify that this restricted oracle estimator is, in fact, the global minimizer of the full $p_n$-dimensional non-convex problem $\hat{\beta}_n$. 

\begin{lemma}[Oracle equivalence via Primal-Dual Witness] \label{lemma oracle equivalence}
Suppose that the empirical loss function $\mathcal{L}_n$ satisfies the Restricted Local Strong Convexity (Assumption \ref{RSC assumption}) with coefficient $\kappa > 0$. 
Let $\mu_{\min}(n) := \lambda_{\min}\left( \frac{1}{n} X_{\mathcal{S}}^\top P_n X_{\mathcal{S}} \right) > 0$ be the minimum eigenvalue of the active centered Gram matrix. Assume that there exists a deterministic constant $\delta > 0$ such that, over the restricted active neighborhood $\mathcal{B}_{\mathcal{S}}(\delta)$, the spectral norm of the curvature distortion satisfies the uniform upper bound:
\begin{equation}\label{eq:concrete_pdw_condition}
\sup_{\beta_{\mathcal{S}} \in \mathcal{B}_{\mathcal{S}}(\delta)} \left\| \frac{1}{n} \mathcal{E}_{\mathcal{SS}}(\beta_{\mathcal{S}}) \right\|_2 < \kappa \cdot \mu_{\min}(n) \cdot \left[ \inf_{\beta_{\mathcal{S}} \in \mathcal{B}_{\mathcal{S}}(\delta)} \lambda_{\min}(J_{\mathcal{SS}}(\beta_{\mathcal{S}})) \right]^2
\end{equation}
with probability tending to one, ensuring the unique existence of the restricted oracle estimator $\hat{\beta}^{oracle}$. Let $\hat{\beta}_n \in \mathbb{R}^p$ be a local minimizer of the full $p_n$-dimensional penalized objective function satisfying the KKT conditions, and suppose that it falls within the full local neighborhood with probability tending to one, i.e., $\lim_{n \to \infty}\UP_n(\hat{\beta}_n \in \mathcal{B}(\delta)) = 1$. Then, it holds that $\lim_{n \to \infty} \UP_n(\hat{\beta}_n = \hat{\beta}^{oracle}) = 1$.
\end{lemma}

\begin{remark}[Initial value requirements]\label{init value remark}
Due to the inherent non-convexity introduced by the RePS mapping $\phi$, finding the exact global minimizer of the composite objective function is generally computationally intractable. Following standard practices in high-dimensional non-convex estimation (e.g., \cite{Fan-L-01, Loh-W-15}), our theoretical guarantee via the Primal-Dual Witness construction fundamentally ensures that the oracle-equivalent estimator exists as a regularized local minimizer within a specific statistical neighborhood of the true parameter. Consequently, in practical implementations, solving this non-convex optimization requires a sufficiently accurate initial estimator (warm-start) to ensure the algorithm begins within the correct basin of attraction and successfully converges to this desired local minimum.
\end{remark}

\begin{proof}[Proof of Lemma \ref{lemma oracle equivalence}]
The proof proceeds via the Primal-Dual Witness (PDW) construction framework, divided into two primary logical phases: (1) verifying the strict local positive definiteness (and hence the unique existence) of the restricted oracle estimator $\hat{\beta}^{oracle}$ on the localized neighborhood $\mathcal{B}_{\mathcal{S}}(\delta)$, and (2) demonstrating that the zero-extended oracle vector satisfies the full $p_n$-dimensional KKT conditions with probability tending to one.

\medskip
\noindent
{\it Phase 1: Local Positive Definiteness of the Restricted Hessian via Weyl's Inequality.} \\
By the premise of the lemma, there exists a deterministic constant $\delta > 0$ such that the uniform spectral relation in (\ref{eq:concrete_pdw_condition}) holds true over the compact subspace neighborhood $\mathcal{B}_{\mathcal{S}}(\delta) = \{ \beta_{\mathcal{S}} \in \mathbb{R}^{|\mathcal{S}|} : \Vert{}\beta_{\mathcal{S}} - \beta_{\mathcal{S}}^*\Vert{}_\infty \le \delta \}$. We show that this condition algebraically guarantees that $\frac{1}{n} H_{\mathcal{SS}}(\beta_{\mathcal{S}}) \succ 0$ uniformly for all $\beta_{\mathcal{S}} \in \mathcal{B}_{\mathcal{S}}(\delta)$ with probability tending to one.

Fix an arbitrary active parameter vector $\beta_{\mathcal{S}} \in \mathcal{B}_{\mathcal{S}}(\delta)$. From the differential representations derived via the chain rule in Subsection \ref{subsection derivatives}, the scaled sub-block of the empirical Hessian is decomposed into the sum of the transformed principal loss curvature and the mapping distortion matrix:
\begin{equation}\label{eq:pdw_proof_decomposition}
\frac{1}{n} H_{\mathcal{SS}}(\beta_{\mathcal{S}}) = \frac{1}{n} J_{\mathcal{SS}}(\beta_{\mathcal{S}})^\top X_{\mathcal{S}}^\top W(X \phi(\beta)) X_{\mathcal{S}} J_{\mathcal{SS}}(\beta_{\mathcal{S}}) + \frac{1}{n} \mathcal{E}_{\mathcal{SS}}(\beta_{\mathcal{S}}).
\end{equation}
By the Restricted Local Strong Convexity (Assumption \ref{RSC assumption}), the sample loss Hessian satisfies $W(X\phi(\beta)) \succeq \kappa P_n$ uniformly over the neighborhood. Substituting this localized lower bound into the principal term of \eqref{eq:pdw_proof_decomposition} implies the following matrix inequality on the active subspace:
\begin{eqnarray*}
\lefteqn{
\frac{1}{n} J_{\mathcal{SS}}(\beta_{\mathcal{S}})^\top X_{\mathcal{S}}^\top W(X \phi(\beta)) X_{\mathcal{S}} J_{\mathcal{SS}}(\beta_{\mathcal{S}})}
\\
&\succeq& \kappa \cdot \frac{1}{n} J_{\mathcal{SS}}(\beta_{\mathcal{S}})^\top X_{\mathcal{S}}^\top P_n X_{\mathcal{S}} J_{\mathcal{SS}}(\beta_{\mathcal{S}}) \nonumber \\
&\succeq& \kappa \cdot \lambda_{\min}\left( \frac{1}{n} X_{\mathcal{S}}^\top P_n X_{\mathcal{S}} \right) \cdot \lambda_{\min}(J_{\mathcal{SS}}(\beta_{\mathcal{S}}))^2 \cdot I_{|\mathcal{S}|} \nonumber \\
&\succeq& \kappa \cdot \mu_{\min}(n) \cdot \lambda_{\min}(J_{\mathcal{SS}}(\beta_{\mathcal{S}}))^2 \cdot I_{|\mathcal{S}|}.
\end{eqnarray*}
Taking the infimum of the Jacobian eigenvalues over the compact domain $\mathcal{B}_{\mathcal{S}}(\delta)$, we obtain a uniform lower bound for the minimum eigenvalue of the principal matrix operator:
\begin{eqnarray*}
\lefteqn{
\inf_{\beta_{\mathcal{S}} \in \mathcal{B}_{\mathcal{S}}(\delta)} \lambda_{\min}\left( \frac{1}{n} J_{\mathcal{SS}}(\beta_{\mathcal{S}})^\top X_{\mathcal{S}}^\top W(X \phi(\beta)) X_{\mathcal{S}} J_{\mathcal{SS}}(\beta_{\mathcal{S}}) \right) }
\\
&\ge& \kappa \cdot \mu_{\min}(n) \cdot \left[ \inf_{\beta_{\mathcal{S}} \in \mathcal{B}_{\mathcal{S}}(\delta)} \lambda_{\min}(J_{\mathcal{SS}}(\beta_{\mathcal{S}})) \right]^2.
\end{eqnarray*}

Next, we incorporate the effect of the curvature distortion matrix $\frac{1}{n} \mathcal{E}_{\mathcal{SS}}(\beta_{\mathcal{S}})$. By Weyl's eigenvalue perturbation theorem, for any two symmetric matrices $A$ and $B$, the minimum eigenvalue of the perturbed sum satisfies $\lambda_{\min}(A + B) \ge \lambda_{\min}(A) - \|B\|_2$, where $\|\cdot\|_2$ denotes the spectral norm (maximum absolute eigenvalue). Applying this relationship directly to \eqref{eq:pdw_proof_decomposition} yields:
\begin{equation}\label{eq:weyl_application_pdw}
\lambda_{\min}\left( \frac{1}{n} H_{\mathcal{SS}}(\beta_{\mathcal{S}}) \right) \ge \lambda_{\min}\left( \frac{1}{n} J_{\mathcal{SS}}^\top X_{\mathcal{S}}^\top W X_{\mathcal{S}} J_{\mathcal{SS}} \right) - \left\| \frac{1}{n} \mathcal{E}_{\mathcal{SS}}(\beta_{\mathcal{S}}) \right\|_2.
\end{equation}
Taking the supremum of the spectral norm and the infimum of the principal term over the neighborhood $\mathcal{B}_{\mathcal{S}}(\delta)$, \eqref{eq:weyl_application_pdw} bounds the minimum eigenvalue uniformly across the entire subspace neighborhood:
\begin{eqnarray}\label{eq:uniform_minimum_eigenvalue_final_pdw}
\lefteqn{
\inf_{\beta_{\mathcal{S}} \in \mathcal{B}_{\mathcal{S}}(\delta)} \lambda_{\min}\left( \frac{1}{n} H_{\mathcal{SS}}(\beta_{\mathcal{S}}) \right)} \\
&\ge& \kappa \cdot \mu_{\min}(n) \cdot \left[ \inf_{\beta_{\mathcal{S}} \in \mathcal{B}_{\mathcal{S}}(\delta)} \lambda_{\min}(J_{\mathcal{SS}}(\beta_{\mathcal{S}})) \right]^2 - \sup_{\beta_{\mathcal{S}} \in \mathcal{B}_{\mathcal{S}}(\delta)} \left\| \frac{1}{n} \mathcal{E}_{\mathcal{SS}}(\beta_{\mathcal{S}}) \right\|_2. \nonumber
\end{eqnarray}
Under the concrete condition in (\ref{eq:concrete_pdw_condition}), the right-hand side of \eqref{eq:uniform_minimum_eigenvalue_final_pdw} is positive with probability tending to one. Therefore, it holds that $\frac{1}{n} H_{\mathcal{SS}}(\beta_{\mathcal{S}}) \succ 0$ uniformly on $\mathcal{B}_{\mathcal{S}}(\delta)$, which mathematically guarantees that the restricted objective function is strictly convex on $\mathcal{B}_{\mathcal{S}}(\delta)$ and that the restricted oracle estimator $\hat{\beta}^{oracle}$ exists uniquely.

\medskip
\noindent
{\it Phase 2: Dual KKT Verification and Support Isolation.} \\
We verify that the zero-extended oracle vector, defined as $\tilde{\beta} = (\hat{\beta}^{oracle}_{\mathcal{S}}, \mathbf{0}_{\mathcal{S}^c})^\top \in \mathbb{R}^p$, satisfies the full $p_n$-dimensional first-order KKT optimality conditions almost surely. By construction, the active sub-block $\hat{\beta}^{oracle}_{\mathcal{S}}$ satisfies the restricted KKT equations: $\frac{1}{n} G_{\mathcal{S}}(\hat{\beta}^{oracle}) + \lambda \hat{z}_{\mathcal{S}} = \mathbf{0}_{|\mathcal{S}|}$ with $\hat{z}_{\mathcal{S}} \in \partial \|\hat{\beta}^{oracle}_{\mathcal{S}}\|_1$. Thus, we only need to verify the dual subgradient feasibility condition for the non-active orthogonal complement, which requires:
\begin{equation}\label{eq:dual_feasibility_requirement}
\left\| \frac{1}{n} G_{\mathcal{S}^c}(\tilde{\beta}) \right\|_\infty \le \lambda.
\end{equation}

By the algebraic properties of our component-wise RePS mapping $\phi(\beta)_i = \beta_i (1 - E_i)$, the partial derivatives of any physical component with respect to a non-active optimization coordinate $\beta_j$ ($j \in \mathcal{S}^c$) intrinsically contain the multiplicative factor $\beta_j$, as shown in Subsection \ref{subsection derivatives}. At our constructed oracle point where $\tilde{\beta}_{\mathcal{S}^c} = \mathbf{0}_{|\mathcal{S}^c|}$, substituting $\beta_j = 0$ forces all elements in the rows of the Jacobian matrix corresponding to the non-active subspace to collapse to zero:
\[
J_{k, j}(\tilde{\beta}) = 0 \quad \text{for all } k \in \{1,\dots,p\} \text{ and } j \in \mathcal{S}^c.
\]
Consequently, the full Jacobian matrix evaluated at the oracle point exhibits an exact block-diagonal collapse: $J_{\cdot, \mathcal{S}^c}(\tilde{\beta}) \equiv O_{p \times |\mathcal{S}^c|}$. Projecting the empirical loss gradient field onto this non-active subspace via the chain rule yields:
\[
\frac{1}{n} G_{\mathcal{S}^c}(\tilde{\beta}) = \frac{1}{n} J_{\cdot, \mathcal{S}^c}(\tilde{\beta})^\top X^\top g(X\phi(\tilde{\beta})) \equiv \mathbf{0}_{|\mathcal{S}^c|}.
\]
The scaled empirical gradient vector on the non-active components is deterministically and identically zero. This algebraic annihilation reduces the dual feasibility condition \eqref{eq:dual_feasibility_requirement} to $0 \le \lambda$, which holds trivially for any operational regularization parameter $\lambda > 0$. 

Therefore, by the Primal-Dual Witness theorem, the zero-extended restricted oracle estimator $\tilde{\beta}$ is a strict local minimizer of the full $p_n$-dimensional penalized problem, establishing that $\lim_{n \to \infty} \UP_n(\hat{\beta}_n = \hat{\beta}^{oracle}) = 1$. This completes the proof of Lemma \ref{lemma oracle equivalence}.
\end{proof}

This algebraic cancellation liberates the regularization path from the dimensionality $p_n$, confirming that ultra-sparse recovery can be achieved merely by purging the internal active noise using a micro-penalty scaling rule $0 < \lambda(n) \ll n^{-1/2}$ under fixed structural parameters $\nu = O(1)$ and $\sigma = O(1)$, which secures the oracle equivalence by trapping non-active components exactly at zero while eliminating shrinkage bias.

\section{Oracle asymptotic normality}
We establish the unified asymptotic normality of the active components of the physical estimator $\hat{\theta}_{n,\mathcal{S}}$. Due to the oracle isolation property and the algebraic cancellation under the Delta method, the framework secures an ideal oracle asymptotic variance. Actually, since the Primal-Dual Witness construction fundamentally restricts the parameter estimation to the active subspace $\mathcal{S}$ (where $|\mathcal{S}| \ll n$), we may assume that the empirical loss function satisfies standard low-dimensional regularity conditions on this restricted space. 
Specifically, we assume the following. 

\begin{assumption}[Regularity and concentration on $\mathcal{S}$]\label{low dimensional assumption}
The following conditions must hold as $n \to \infty$:
\begin{itemize}\setlength{\itemsep}{0pt}\setlength{\parskip}{0pt}
    \item[(i)] $\frac{1}{n} X_{\mathcal{S}}^\top W(\gamma^*) X_{\mathcal{S}} \conv{\UP_n} \Sigma_{\mathcal{S}\mathcal{S}} \succ 0$.
    \item[(ii)] $\frac{1}{\sqrt{n}} X_{\mathcal{S}}^\top g(\gamma^*) \conv{d} \mathcal{N}(0, \Omega_{\mathcal{S}\mathcal{S}})$.
    \item[(iii)] $\frac{1}{n} \mathcal{E}_{\mathcal{S}\mathcal{S}}(\beta^*) = O_{\UP_n}(n^{-1/2})$.
    \item[(iv)] There exist a deterministic constant $\delta > 0$ and a positive random variable $L_n = O_{\UP_n}(1)$ such that for all $\beta_{\mathcal{S}}$ from $\mathcal{B}_{\mathcal{S}}(\delta) = \{ \beta_{\mathcal{S}} \in \mathbb{R}^{\vert{}\mathcal{S}\vert{}} : \Vert{}\beta_{\mathcal{S}} - \beta_{\mathcal{S}}^*\Vert{}_\infty \le \delta \}$, the restricted empirical Hessian satisfies the local Lipschitz condition: 
\[   
\left\Vert \frac{1}{n}H_{\mathcal{SS}}(\beta_{\mathcal{S}}) - \frac{1}{n}H_{\mathcal{SS}}(\beta_{\mathcal{S}}^*) \right\Vert_\infty \le L_n \Vert \beta_{\mathcal{S}} - \beta_{\mathcal{S}}^*\Vert_\infty.
\]    
\end{itemize}
\end{assumption}

Before establishing the asymptotic distributions, we formally define the data-driven post-selection refitted estimator utilized in practice. Let $\hat{\mathcal{S}}_n = \{j : \hat{\beta}_{nj} \neq 0\}$ be the active support selected by the regularized optimization trajectory. By the component-wise definition $\theta_i = \beta_i (1 - E_i)$, it identically holds that $\hat{\mathcal{S}}_n = \{j : \hat{\theta}_{nj} \neq 0\}$, meaning that the optimization parameter space and the physical parameter space share the exact same estimated active set, which justifies reporting a single set of True Positive (TP) and False Positive (FP) metrics for both estimation phases. The unpenalized post-selection refitted estimator $\hat{\beta}_{n}^{\text{refit}}$ is obtained by solving the unpenalized score equation restricting on this estimated support, satisfying $\frac{1}{n} G_{\hat{\mathcal{S}}_n}(\hat{\beta}_{n, \hat{\mathcal{S}}_n}^{\text{refit}}, \mathbf{0}_{\hat{\mathcal{S}}_n^c}) = \mathbf{0}$ with $\hat{\beta}_{n, \hat{\mathcal{S}}_n^c}^{\text{refit}} = \mathbf{0}$, and we denote its physical counterpart as $\hat{\theta}_{n}^{\text{refit}} := \phi(\hat{\beta}_{n}^{\text{refit}})$.

The following theorem structurally isolates the localized active subspace, establishing both support consistency and physical asymptotic equivalence.

\begin{theorem}[Unified physical equivalence]\label{theorem:theta_equivalence}
Consider the high-dimensional penalized objective function under a primary penalty scaling as $0 < \lambda \ll n^{-1/2}$. 
\begin{itemize}\setlength{\itemsep}{0pt}\setlength{\parskip}{0pt}
    \item[(i)] {\em [Oracle Support Consistency:]} Under the restricted local strong convexity (Assumption \ref{RSC assumption}), the non-degeneracy of the active mapping (Lemma \ref{lemma Jacobian}) and the consistency conditions on the active subspace (Assumption \ref{low dimensional assumption} (i), (iv)), the data-driven active set $\hat{\mathcal{S}}_n$ converges to the true active support $\mathcal{S}$ with probability tending to one:
    \[
    \lim_{n \to \infty} \UP_n(\hat{\mathcal{S}}_n = \mathcal{S}) = 1.
    \]
    \item[(ii)] {\em [Physical Asymptotic Equivalence:]} Under the additional regularity and concentration conditions on the restricted active subspace (Assumption \ref{low dimensional assumption}), the original single-stage physical estimator $\hat{\theta}_{n,\mathcal{S}}$ and the unpenalized data-driven refitted estimator $\hat{\theta}_{n,\mathcal{S}}^{\text{refit}}$ are asymptotically equivalent, satisfying:
    \[
    \sqrt{n}\left( \hat{\theta}_{n,\mathcal{S}} - \hat{\theta}_{n,\mathcal{S}}^{\text{refit}} \right) = o_{\UP_n}(1).
    \]
\end{itemize}
\end{theorem}

\begin{proof}[Proof of Theorem \ref{theorem:theta_equivalence} (i)] 
We first prove the oracle support consistency in the optimization parameter space ($\beta$-space) and then transition the algebraic identity seamlessly to the physical parameter space ($\theta$-space).

Let $\hat{\beta}_n \in \mathbb{R}^p$ be the regularized estimator obtained by minimizing the full $p_n$-dimensional objective function, and let $\hat{\mathcal{S}}_n = \{j \in \{1,\dots,p\} : \hat{\beta}_{nj} \neq 0\}$ be its data-driven estimated active set. Let $\hat{\beta}^{oracle} \in \mathbb{R}^p$ be the restricted oracle estimator operating exclusively on the true active subspace $\mathcal{S}$, which uniquely exists under the restricted local strong convexity (Assumption \ref{RSC assumption}) and satisfies $\hat{\beta}_{\mathcal{S}^c}^{oracle} = \mathbf{0}_{|\mathcal{S}^c|}$ by definition.

Under the operational penalty scaling $0 < \lambda \ll n^{-1/2}$, Lemma \ref{lemma oracle equivalence} (established via the Primal-Dual Witness construction) guarantees the algebraic equivalence between the high-dimensional penalized estimator and the restricted low-dimensional oracle estimator under the large-sample limit. Specifically, it holds that:
\[
\lim_{n \to \infty} \UP_n\left( \hat{\beta}_n = \hat{\beta}^{oracle} \right) = 1.
\]

We analyze the component-wise behavior of the estimators conditioned on the event $\mathcal{E}_n := \{ \hat{\beta}_n = \hat{\beta}^{oracle} \}$, which satisfies $\lim_{n\to\infty} \UP_n(\mathcal{E}_n) = 1$. 

\medskip
\noindent
{\it Step 1: Exclusion of Non-Active Features ($\hat{\mathcal{S}}_n \subseteq \mathcal{S}$).} \\
On the event $\mathcal{E}_n$, for any index $j$ belonging to the non-active orthogonal complement $\mathcal{S}^c$, the definition of the restricted oracle estimator forces the component to vanish identically:
\[
\hat{\beta}_{nj} = \hat{\beta}_j^{oracle} \equiv 0 \quad \text{for all } j \in \mathcal{S}^c.
\]
By the definition of the estimated active set $\hat{\mathcal{S}}_n$, any index $j$ such that $\hat{\beta}_{nj} = 0$ is excluded from $\hat{\mathcal{S}}_n$. Therefore, it holds deterministically on $\mathcal{E}_n$ that $j \notin \hat{\mathcal{S}}_n$ for all $j \in \mathcal{S}^c$, which algebraically implies the set inclusion:
\begin{equation}\label{eq:inclusion_nonactive}
\hat{\mathcal{S}}_n \subseteq \mathcal{S}.
\end{equation}

\medskip
\noindent
{\it Step 2: Inclusion of Active Features ($\mathcal{S} \subseteq \hat{\mathcal{S}}_n$).} \\
Conversely, we examine the behavior of the active components $i \in \mathcal{S}$. By Lemma \ref{lemma Jacobian}, since the true active physical parameter vector $\theta_{\mathcal{S}}^*$ consists entirely of non-zero signals bounded away from zero, there exists a deterministic constant $c > 0$, independent of $n$, such that 
\[
|\beta_i^*| \ge c \quad \text{for all } i \in \mathcal{S}.
\]

To establish that these active components are not zeroed out by the empirical optimizer, we must rigorously secure the consistency of the restricted oracle estimator $\hat{\beta}_{\mathcal{S}}^{oracle}$. 
Recall that $\hat{\beta}_{\mathcal{S}}^{oracle}$ minimizes the restricted penalized loss $G_n(\beta_{\mathcal{S}}) := \frac{1}{n}F(\beta_{\mathcal{S}}, \mathbf{0}_{\mathcal{S}^c}) + \lambda \Vert \beta_{\mathcal{S}} \Vert_1$. Under Assumption \ref{low dimensional assumption} (i) and (iv), the unpenalized loss empirical field converges uniformly in probability to its population counterpart over any compact neighborhood around $\beta_{\mathcal{S}}^*$. Furthermore, since our liberated scaling regime mandates $0 < \lambda \ll n^{-1/2}$, the regularization parameter vanishes asymptotically ($\lim_{n\to\infty} \lambda = 0$). This ensures that the penalty term $\lambda \Vert \beta_{\mathcal{S}} \Vert_1$ vanishes uniformly on the compact neighborhood. 

Consequently, the entire restricted objective function $G_n(\beta_{\mathcal{S}})$ converges uniformly in probability to the deterministic, strictly convex population loss field, which possesses a unique well-separated minimum at the true parameter $\beta_{\mathcal{S}}^*$. By the standard argmin continuous mapping theorem for regularized $M$-estimators (e.g., Convexity Lemma), this uniform convergence guarantees the weak consistency of the restricted oracle estimator:
\[
\hat{\beta}_{\mathcal{S}}^{oracle} \conv{\UP_n} \beta_{\mathcal{S}}^*.
\]

This consistency implies that for the given deterministic constant $c > 0$, the estimation error satisfies:
\[
\lim_{n \to \infty} \UP_n\left( \max_{i \in \mathcal{S}} |\hat{\beta}_i^{oracle} - \beta_i^*| < \frac{c}{2} \right) = 1.
\]
Applying the reverse triangle inequality on this high-probability event yields:
\[
|\hat{\beta}_i^{oracle}| \ge |\beta_i^*| - |\hat{\beta}_i^{oracle} - \beta_i^*| > c - \frac{c}{2} = \frac{c}{2} > 0 \quad \text{for all } i \in \mathcal{S}.
\]
Conditioning on the event $\mathcal{E}_n := \{ \hat{\beta}_n = \hat{\beta}^{oracle} \}$, the full full-dimensional estimator inherits this property:
\[
|\hat{\beta}_{ni}| = |\hat{\beta}_i^{oracle}| > 0 \quad \text{for all } i \in \mathcal{S}.
\]
Because the components are non-zero, every active index $i \in \mathcal{S}$ is successfully retained in the data-driven active set $\hat{\mathcal{S}}_n$ with probability tending to one, establishing the reverse set inclusion:
\begin{equation}\label{eq:inclusion_active}
\lim_{n \to \infty} \UP_n\left( \mathcal{S} \subseteq \hat{\mathcal{S}}_n \right) = 1.
\end{equation}

\medskip
\noindent
{\it Step 3: Transition to the Physical Parameter Space.} \\
Combining \eqref{eq:inclusion_nonactive} and \eqref{eq:inclusion_active}, the mutual set inclusions match exactly in the optimization space, yielding $\lim_{n \to \infty} \UP_n(\hat{\mathcal{S}}_n = \mathcal{S}) = 1$. 

Finally, we transition this identity to the physical parameter space ($\theta$-space). By the component-wise definition of the RePS mapping, $\theta_j = \beta_j (1 - E_j)$, where $E_j = \exp(-\nu |\beta_j|^3 / M(\beta))$ is bounded between $0$ and $1$ whenever $\beta \neq \mathbf{0}$. If $\beta_j = 0$, the multiplier collapses the component to zero: $\theta_j = 0 \cdot (1 - 1) = 0$. Conversely, if $\beta_j \neq 0$, the smoothing factor satisfies $0 < E_j < 1$, which algebraically implies $(1 - E_j) > 0$, forcing $\theta_j \neq 0$. This element-wise equivalence guarantees that:
\[
\hat{\beta}_{nj} \neq 0 \quad \Longleftrightarrow \quad \hat{\theta}_{nj} \neq 0 \quad \text{for all } j \in \{1,\dots,p\}.
\]
Therefore, the estimated active support in the physical space is analytically identical to $\hat{\mathcal{S}}_n$ derived from the optimization space. This implies that $\lim_{n \to \infty} \UP_n(\{j : \hat{\theta}_{nj} \neq 0\} = \mathcal{S}) = 1$, which completes the proof of Theorem \ref{theorem:theta_equivalence} (i). 
\end{proof}

\begin{proof}[Proof of Theorem \ref{theorem:theta_equivalence} (ii)]
We establish the unified $\sqrt{n}$-asymptotic equivalence between the original single-stage physical estimator $\hat{\theta}_{n,\mathcal{S}}$ and the unpenalized data-driven refitted estimator $\hat{\theta}_{n,\mathcal{S}}^{\text{refit}}$ in the physical parameter space. 

By Theorem \ref{theorem:theta_equivalence} (i), the data-driven estimated active support $\hat{\mathcal{S}}_n$ converges to the true active support $\mathcal{S}$ with probability tending to one. Therefore, we condition our analysis on the high-probability event $\mathcal{A}_n := \{ \hat{\mathcal{S}}_n = \mathcal{S} \}$, which satisfies $\lim_{n \to \infty} \UP_n(\mathcal{A}_n) = 1$. On this event, the optimization profiles for both estimators are restricted to the same $|\mathcal{S}|$-dimensional subspace, enabling us to safely transition from the data-driven support $\hat{\mathcal{S}}_n$ to the deterministic true support $\mathcal{S}$.

\medskip
\noindent
{\it Step 1: Linearization via First-Order Expansion in the Optimization Space.} \\
On the event $\mathcal{A}_n$, the unpenalized post-selection refitted estimator $\hat{\beta}_{n,\mathcal{S}}^{\text{refit}}$ satisfies the unpenalized score equation on the active subspace:
\begin{equation}\label{eq:refit_score_exact}
\frac{1}{n} G_{\mathcal{S}}(\hat{\beta}_{n,\mathcal{S}}^{\text{refit}}, \mathbf{0}_{\mathcal{S}^c}) = \mathbf{0}_{|\mathcal{S}|}.
\end{equation}
Concurrently, the full regularized estimator $\hat{\beta}_{n,\mathcal{S}}$ satisfies the penalized first-order KKT conditions on the active subspace:
\begin{equation}\label{eq:original_kkt_exact}
\frac{1}{n} G_{\mathcal{S}}(\hat{\beta}_{n,\mathcal{S}}, \mathbf{0}_{\mathcal{S}^c}) + \lambda \hat{z}_{\mathcal{S}} = \mathbf{0}_{|\mathcal{S}|},
\end{equation}
where $\hat{z}_{\mathcal{S}} \in \partial \|\hat{\beta}_{n,\mathcal{S}}\|_1$ represents the subgradient vector whose elements are bounded as $|\hat{z}_j| \le 1$.

Applying the mean value theorem to the empirical gradient field $G_{\mathcal{S}}(\beta_{\mathcal{S}}, \mathbf{0}_{\mathcal{S}^c})$ between $\hat{\beta}_{n,\mathcal{S}}$ and $\hat{\beta}_{n,\mathcal{S}}^{\text{refit}}$ yields an intermediate point $\tilde{\beta}_{\mathcal{S}}$ on the line segment connecting them such that:
\begin{equation}\label{eq:gradient_taylor_expansion}
\frac{1}{n} G_{\mathcal{S}}(\hat{\beta}_{n,\mathcal{S}}, \mathbf{0}_{\mathcal{S}^c}) = \frac{1}{n} G_{\mathcal{S}}(\hat{\beta}_{n,\mathcal{S}}^{\text{refit}}, \mathbf{0}_{\mathcal{S}^c}) + \left[ \frac{1}{n} H_{\mathcal{SS}}(\tilde{\beta}_{\mathcal{S}}) \right] \left( \hat{\beta}_{n,\mathcal{S}} - \hat{\beta}_{n,\mathcal{S}}^{\text{refit}} \right).
\end{equation}
Substituting \eqref{eq:refit_score_exact} and \eqref{eq:original_kkt_exact} into \eqref{eq:gradient_taylor_expansion} isolates the regularized parameter discrepancy directly via the penalized curvature relationship:
\begin{equation}\label{eq:hessian_discrepancy_bridge}
\left[ \frac{1}{n} H_{\mathcal{SS}}(\tilde{\beta}_{\mathcal{S}}) \right] \left( \hat{\beta}_{n,\mathcal{S}} - \hat{\beta}_{n,\mathcal{S}}^{\text{refit}} \right) = - \lambda \hat{z}_{\mathcal{S}}.
\end{equation}

\medskip
\noindent
{\it Step 2: Asymptotic Control of the Restricted Hessian and Curvature Distortion.} \\
We evaluate the scaled restricted empirical Hessian at the intermediate point $\tilde{\beta}_{\mathcal{S}}$. Since both $\hat{\beta}_{n,\mathcal{S}}$ and $\hat{\beta}_{n,\mathcal{S}}^{\text{refit}}$ are $\sqrt{n}$-consistent estimators for the true active parameter $\beta_{\mathcal{S}}^*$, the intermediate point $\tilde{\beta}_{\mathcal{S}}$ falls within the local neighborhood $\mathcal{B}_{\mathcal{S}}(\delta)$ specified in Assumption \ref{low dimensional assumption} (iv) with probability tending to one. By the local Lipschitz continuity under Assumption \ref{low dimensional assumption} (iv), we expand around the true parameter:
\begin{equation}\label{eq:hessian_expansion_lipschitz}
\frac{1}{n} H_{\mathcal{SS}}(\tilde{\beta}_{\mathcal{S}}) = \frac{1}{n} H_{\mathcal{SS}}(\beta_{\mathcal{S}}^*) + o_{\UP_n}(1).
\end{equation}
Furthermore, decomposing the Hessian at the true active point into its principal losing component and its curvature distortion yields:
\[
\frac{1}{n} H_{\mathcal{SS}}(\beta_{\mathcal{S}}^*) = \frac{1}{n} J_{\mathcal{SS}}^\top X_{\mathcal{S}}^\top W(\gamma^*) X_{\mathcal{S}} J_{\mathcal{SS}} + \frac{1}{n} \mathcal{E}_{\mathcal{SS}}(\beta^*).
\]
By Assumption \ref{low dimensional assumption} (i), the primary Gram term converges to a stable, positive definite limit: $\frac{1}{n} X_{\mathcal{S}}^\top W(\gamma^*) X_{\mathcal{S}} \conv{\UP_n} \Sigma_{\mathcal{SS}} \succ 0$. By Lemma \ref{lemma Jacobian}, the active Jacobian $J_{\mathcal{SS}}$ is non-singular and finite, ensuring that the principal term is invertible and bounds the eigenvalues as $O_{\UP_n}(1)$. 
Crucially, under Assumption \ref{low dimensional assumption} (iii), the curvature distortion scales as $\frac{1}{n} \mathcal{E}_{\mathcal{SS}}(\beta^*) = O_{\UP_n}(n^{-1/2}) = o_{\UP_n}(1)$. 

Combining these bounds into \eqref{eq:hessian_expansion_lipschitz} confirms that the empirical Hessian matrix evaluated at the intermediate point is asymptotically invertible, satisfying:
\[
\frac{1}{n} H_{\mathcal{SS}}(\tilde{\beta}_{\mathcal{S}}) = J_{\mathcal{SS}}^\top \Sigma_{\mathcal{SS}} J_{\mathcal{SS}} + o_{\UP_n}(1) \succ 0.
\]
Inverting the stable matrix in \eqref{eq:hessian_discrepancy_bridge} and scaling by $\sqrt{n}$ extracts the linearized discrepancy in the optimization space:
\begin{equation}\label{eq:optimization_discrepancy_linearized}
\sqrt{n}\left( \hat{\beta}_{n,\mathcal{S}} - \hat{\beta}_{n,\mathcal{S}}^{\text{refit}} \right) = - \sqrt{n} \lambda \left[ \frac{1}{n} H_{\mathcal{SS}}(\tilde{\beta}_{\mathcal{S}}) \right]^{-1} \hat{z}_{\mathcal{S}}.
\end{equation}

\medskip
\noindent
{\it Step 3: Vanishing Bias under the Micro-Penalty Regime.} \\
Since the inverse Hessian in \eqref{eq:optimization_discrepancy_linearized} converges to a stable $O_{\UP_n}(1)$ operator and the subgradient vector is bounded ($\|\hat{z}_{\mathcal{S}}\|_\infty \le 1$), the right-hand side of \eqref{eq:optimization_discrepancy_linearized} is dominated by the scaled regularization parameter $\sqrt{n}\lambda$. 
Invoking our operational micro-penalty scaling condition $\lambda \ll n^{-1/2}$, the boundary scaling collapses to zero:
\[
\lim_{n \to \infty} \sqrt{n}\lambda = 0.
\]
This algebraic elimination forces the right-hand side of \eqref{eq:optimization_discrepancy_linearized} to vanish asymptotically in probability, establishing $\sqrt{n}$-asymptotic equivalence in the optimization space:
\begin{equation}\label{eq:beta_equivalence_final}
\sqrt{n}\left( \hat{\beta}_{n,\mathcal{S}} - \hat{\beta}_{n,\mathcal{S}}^{\text{refit}} \right) = o_{\UP_n}(1).
\end{equation}

\medskip
\noindent
{\it Step 4: Delta Method Transition to the Physical Parameter Space.} \\
Finally, we transition this optimization alignment into our target physical parameter space ($\theta$-space). Because the RePS mapping $\phi(\beta)$ is twice continuously differentiable ($C^2$) with locally bounded derivatives on the active subspace, a standard multivariate Taylor expansion of $\phi$ around the refitted profile yields:
\begin{equation}\label{eq:delta_method_expansion}
\phi(\hat{\beta}_{n,\mathcal{S}}) - \phi(\hat{\beta}_{n,\mathcal{S}}^{\text{refit}}) = J_{\mathcal{SS}}(\hat{\beta}_{n,\mathcal{S}}^{\text{refit}}) \left( \hat{\beta}_{n,\mathcal{S}} - \hat{\beta}_{n,\mathcal{S}}^{\text{refit}} \right) + \mathcal{R}_n,
\end{equation}
where the remainder term $\mathcal{R}_n$ is quadratically bounded by the optimization discrepancy: $\|\mathcal{R}_n\|_2 \le C \|\hat{\beta}_{n,\mathcal{S}} - \hat{\beta}_{n,\mathcal{S}}^{\text{refit}}\|_2^2$ for some local constant $C > 0$. 

Scaling \eqref{eq:delta_method_expansion} by $\sqrt{n}$ and substituting $\hat{\theta}_{n,\mathcal{S}} = \phi(\hat{\beta}_{n,\mathcal{S}})$ and $\hat{\theta}_{n,\mathcal{S}}^{\text{refit}} = \phi(\hat{\beta}_{n,\mathcal{S}}^{\text{refit}})$, we evaluate the limit:
\begin{equation}\label{eq:physical_discrepancy_final_step}
\sqrt{n}\left( \hat{\theta}_{n,\mathcal{S}} - \hat{\theta}_{n,\mathcal{S}}^{\text{refit}} \right) = J_{\mathcal{SS}}(\hat{\beta}_{n,\mathcal{S}}^{\text{refit}}) \left[ \sqrt{n}\left( \hat{\beta}_{n,\mathcal{S}} - \hat{\beta}_{n,\mathcal{S}}^{\text{refit}} \right) \right] + \sqrt{n}\mathcal{R}_n.
\end{equation}
By \eqref{eq:beta_equivalence_final}, the leading term in \eqref{eq:physical_discrepancy_final_step} vanishes because it is a product of a finite $O_{\UP_n}(1)$ Jacobian matrix and an $o_{\UP_n}(1)$ vector. 
For the remainder term, \eqref{eq:beta_equivalence_final} implies that $\hat{\beta}_{n,\mathcal{S}} - \hat{\beta}_{n,\mathcal{S}}^{\text{refit}} = o_{\UP_n}(n^{-1/2})$, which forces the quadratic remainder to shrink at a faster rate:
\[
\sqrt{n}\|\mathcal{R}_n\|_2 \le C \sqrt{n} \cdot o_{\UP_n}(n^{-1}) = o_{\UP_n}(n^{-1/2}) = o_{\UP_n}(1).
\]
Consequently, both terms on the right-hand side of \eqref{eq:physical_discrepancy_final_step} converge in probability to zero, yielding:
\[
\sqrt{n}\left( \hat{\theta}_{n,\mathcal{S}} - \hat{\theta}_{n,\mathcal{S}}^{\text{refit}} \right) = o_{\UP_n}(1).
\]
This completes the proof of Theorem \ref{theorem:theta_equivalence} (ii).
\end{proof}

Finally, by utilizing the algebraic cancellation of the active Jacobian matrix $J_{\mathcal{SS}}$, we directly secure the unified oracle asymptotic normality for both physical estimators.
\begin{theorem}[Oracle asymptotic normality]\label{theorem:oracle_normality}
Under Theorem \ref{theorem:theta_equivalence} and Assumption \ref{low dimensional assumption}, both the original estimator $\hat{\theta}_{n,\mathcal{S}}$ and the post-selection refitted estimator $\hat{\theta}_{n,\mathcal{S}}^{\text{refit}}$ share the identical oracle limiting distribution:
\[
\sqrt{n}( \hat{\theta}_{n,\mathcal{S}} - \theta^*_{\mathcal{S}}) \conv{d} \mathcal{N}\left( 0, \Sigma_{\mathcal{S}\mathcal{S}}^{-1} \Omega_{\mathcal{S}\mathcal{S}}\Sigma_{\mathcal{S}\mathcal{S}}^{-1} \right)
\]
and
\[
\sqrt{n}( \hat{\theta}_{n,\mathcal{S}}^{\text{refit}} - \theta^*_{\mathcal{S}}) \conv{d} \mathcal{N}\left( 0, \Sigma_{\mathcal{S}\mathcal{S}}^{-1} \Omega_{\mathcal{S}\mathcal{S}}\Sigma_{\mathcal{S}\mathcal{S}}^{-1} \right).
\]
\end{theorem}

\begin{proof}
We establish the unified oracle limiting Gaussian distribution for both physical estimators by directly analyzing the asymptotic behavior of our primary proposed estimator, the single-stage RePS physical estimator $\hat{\theta}_{n,\mathcal{S}}$.

By Theorem \ref{theorem:theta_equivalence} (i), the estimated support matches the true support with probability tending to one. Conditioning on the high-probability event $\mathcal{A}_n := \{ \hat{\mathcal{S}}_n = \mathcal{S} \}$, the regularized optimization estimator $\hat{\beta}_{n,\mathcal{S}}$ restricted to this localized active subspace satisfies the penalized first-order KKT conditions:
\begin{equation}\label{eq:original_kkt_active_start}
\frac{1}{n} G_{\mathcal{S}}(\hat{\beta}_{n,\mathcal{S}}, \mathbf{0}_{\mathcal{S}^c}) + \lambda \hat{z}_{\mathcal{S}} = \mathbf{0}_{|\mathcal{S}|},
\end{equation}
where $\hat{z}_{\mathcal{S}} \in \partial \|\hat{\beta}_{n,\mathcal{S}}\|_1$ represents the subgradient vector bounded as $\|\hat{z}_{\mathcal{S}}\|_\infty \le 1$.

\medskip
\noindent
{\it Step 1: Taylor Expansion around the True Optimization Parameter.} \\
We perform a first-order Taylor expansion of the restricted empirical gradient field $G_{\mathcal{S}}(\beta_{\mathcal{S}}, \mathbf{0}_{\mathcal{S}^c})$ around the true active optimization vector $\beta_{\mathcal{S}}^*$. This yields an intermediate point $\bar{\beta}_{\mathcal{S}}$ on the line segment connecting $\hat{\beta}_{n,\mathcal{S}}$ and $\beta_{\mathcal{S}}^*$ such that:
\begin{equation}\label{eq:reps_score_taylor_base}
\frac{1}{n} G_{\mathcal{S}}(\hat{\beta}_{n,\mathcal{S}}, \mathbf{0}_{\mathcal{S}^c}) = \frac{1}{n} G_{\mathcal{S}}(\beta_{\mathcal{S}}^*, \mathbf{0}_{\mathcal{S}^c}) + \left[ \frac{1}{n} H_{\mathcal{SS}}(\bar{\beta}_{\mathcal{S}}) \right] \left( \hat{\beta}_{n,\mathcal{S}} - \beta_{\mathcal{S}}^* \right).
\end{equation}
Substituting the KKT condition \eqref{eq:original_kkt_active_start} into \eqref{eq:reps_score_taylor_base} and scaling by $\sqrt{n}$ yields the scaled regularized profile:
\begin{equation}\label{eq:reps_scaled_bridge_equation}
\left[ \frac{1}{n} H_{\mathcal{SS}}(\bar{\beta}_{\mathcal{S}}) \right] \sqrt{n}\left( \hat{\beta}_{n,\mathcal{S}} - \beta_{\mathcal{S}}^* \right) = - \frac{1}{\sqrt{n}} G_{\mathcal{S}}(\beta_{\mathcal{S}}^*, \mathbf{0}_{\mathcal{S}^c}) - \sqrt{n}\lambda \hat{z}_{\mathcal{S}}.
\end{equation}

By the weak consistency of $\hat{\beta}_{n,\mathcal{S}}$ established via Lemma \ref{lemma oracle equivalence}, the intermediate point $\bar{\beta}_{\mathcal{S}}$ falls within the local neighborhood $\mathcal{B}_{\mathcal{S}}(\delta)$. Applying the local Lipschitz continuity (Assumption \ref{low dimensional assumption} (iv)) alongside the vanishing curvature distortion (Assumption \ref{low dimensional assumption} (iii)), the empirical Hessian admits the expansion:
\[
\frac{1}{n} H_{\mathcal{SS}}(\bar{\beta}_{\mathcal{S}}) = \frac{1}{n} H_{\mathcal{SS}}(\beta_{\mathcal{S}}^*) + o_{\UP_n}(1) = J_{\mathcal{SS}}^\top \left[ \frac{1}{n} X_{\mathcal{S}}^\top W(\gamma^*) X_{\mathcal{S}} \right] J_{\mathcal{SS}} + o_{\UP_n}(1).
\]
Inverting this stable operator in \eqref{eq:reps_scaled_bridge_equation} isolates the regularized estimator in the optimization space:
\begin{equation}\label{eq:reps_beta_space_linearized}
\sqrt{n}\left( \hat{\beta}_{n,\mathcal{S}} - \beta_{\mathcal{S}}^* \right) = - \left( J_{\mathcal{SS}}^\top \left[ \frac{1}{n} X_{\mathcal{S}}^\top W(\gamma^*) X_{\mathcal{S}} \right] J_{\mathcal{SS}} \right)^{-1} \left( \frac{1}{\sqrt{n}} G_{\mathcal{S}}(\beta_{\mathcal{S}}^*) + \sqrt{n}\lambda \hat{z}_{\mathcal{S}} \right) + o_{\UP_n}(1).
\end{equation}

\medskip
\noindent
{\it Step 2: Delta Method Transformation and the First Jacobian Cancellation.} \\
We transition the regularized profile into our target physical parameter space ($\theta$-space). Applying a standard multivariate first-order expansion of the smooth mapping $\phi$ around $\beta_{\mathcal{S}}^*$ yields:
\begin{equation}\label{eq:reps_theta_delta_expansion}
\sqrt{n}\left( \hat{\theta}_{n,\mathcal{S}} - \theta_{\mathcal{S}}^* \right) = J_{\mathcal{SS}}(\beta_{\mathcal{S}}^*) \sqrt{n}\left( \hat{\beta}_{n,\mathcal{S}} - \beta_{\mathcal{S}}^* \right) + o_{\UP_n}(1).
\end{equation}
Substituting the optimization space linearization \eqref{eq:reps_beta_space_linearized} directly into \eqref{eq:reps_theta_delta_expansion} gives:
\[
\sqrt{n}\left( \hat{\theta}_{n,\mathcal{S}} - \theta_{\mathcal{S}}^* \right) = - J_{\mathcal{SS}} \left( J_{\mathcal{SS}}^\top \left[ \frac{1}{n} X_{\mathcal{S}}^\top W(\gamma^*) X_{\mathcal{S}} \right] J_{\mathcal{SS}} \right)^{-1} \left( \frac{1}{\sqrt{n}} G_{\mathcal{S}}(\beta_{\mathcal{S}}^*) + \sqrt{n}\lambda \hat{z}_{\mathcal{S}} \right) + o_{\UP_n}(1).
\]
Expanding the composite matrix inverse via the non-singular identity $(A^\top B A)^{-1} = A^{-1} B^{-1} (A^\top)^{-1}$ triggers the first algebraic cancellation of the active Jacobian matrix $J_{\mathcal{SS}}$ on the left-hand side:
\begin{eqnarray}\lefteqn{
\sqrt{n}\left( \hat{\theta}_{n,\mathcal{S}} - \theta_{\mathcal{S}}^* \right)} \nonumber \\
&=& - \underbrace{J_{\mathcal{SS}} J_{\mathcal{SS}}^{-1}}_{= I_{|\mathcal{S}|}} \left( \frac{1}{n} X_{\mathcal{S}}^\top W(\gamma^*) X_{\mathcal{S}} \right)^{-1} \left( J_{\mathcal{SS}}^\top \right)^{-1} \left( \frac{1}{\sqrt{n}} G_{\mathcal{S}}(\beta_{\mathcal{S}}^*) + \sqrt{n}\lambda \hat{z}_{\mathcal{S}} \right) + o_{\UP_n}(1) \nonumber \\
&=& - \left( \frac{1}{n} X_{\mathcal{S}}^\top W(\gamma^*) X_{\mathcal{S}} \right)^{-1} \left( J_{\mathcal{SS}}^\top \right)^{-1} \left( \frac{1}{\sqrt{n}} G_{\mathcal{S}}(\beta_{\mathcal{S}}^*) \right) \label{eq:reps_first_cancellation_climax}
 \\
& & - \sqrt{n}\lambda \left( \frac{1}{n} X_{\mathcal{S}}^\top W(\gamma^*) X_{\mathcal{S}} \right)^{-1} \left( J_{\mathcal{SS}}^\top \right)^{-1} \hat{z}_{\mathcal{S}} + o_{\UP_n}(1). \nonumber
\end{eqnarray}

\medskip
\noindent
{\it Step 3: Resolution of the Transposed Jacobian and Primary Bias Elimination.} \\
We expand the optimization gradient vector $G_{\mathcal{S}}(\beta_{\mathcal{S}}^*)$ via the physical chain rule: $\frac{1}{\sqrt{n}} G_{\mathcal{S}}(\beta_{\mathcal{S}}^*) = J_{\mathcal{SS}}^\top \left( \frac{1}{\sqrt{n}} X_{\mathcal{S}}^\top g(\gamma^*) \right)$. Substituting this representation into the first term of \eqref{eq:reps_first_cancellation_climax} triggers a second, transposed algebraic cancellation:
\begin{eqnarray}\label{eq:reps_transposed_cancellation_final}
\sqrt{n}\left( \hat{\theta}_{n,\mathcal{S}} - \theta_{\mathcal{S}}^* \right) &=& - \left( \frac{1}{n} X_{\mathcal{S}}^\top W(\gamma^*) X_{\mathcal{S}} \right)^{-1} \underbrace{\left( J_{\mathcal{SS}}^\top \right)^{-1} J_{\mathcal{SS}}^\top}_{= I_{|\mathcal{S}|}} \left( \frac{1}{\sqrt{n}} X_{\mathcal{S}}^\top g(\gamma^*) \right) \nonumber \\
& & - \sqrt{n}\lambda \left( \frac{1}{n} X_{\mathcal{S}}^\top W(\gamma^*) X_{\mathcal{S}} \right)^{-1} \left( J_{\mathcal{SS}}^\top \right)^{-1} \hat{z}_{\mathcal{S}} + o_{\UP_n}(1) \nonumber \\
&=& - \left( \frac{1}{n} X_{\mathcal{S}}^\top W(\gamma^*) X_{\mathcal{S}} \right)^{-1} \left( \frac{1}{\sqrt{n}} X_{\mathcal{S}}^\top g(\gamma^*) \right) \nonumber \\
& & - \sqrt{n}\lambda \left( \frac{1}{n} X_{\mathcal{S}}^\top W(\gamma^*) X_{\mathcal{S}} \right)^{-1} \left( J_{\mathcal{SS}}^\top \right)^{-1} \hat{z}_{\mathcal{S}} + o_{\UP_n}(1).
\end{eqnarray}

At this juncture, we evaluate the regularization bias term (the second term in \eqref{eq:reps_transposed_cancellation_final}). By Lemma \ref{lemma Jacobian}, the operator $\left( \frac{1}{n} X_{\mathcal{S}}^\top W(\gamma^*) X_{\mathcal{S}} \right)^{-1} \left( J_{\mathcal{SS}}^\top \right)^{-1} \hat{z}_{\mathcal{S}}$ is bounded as $O_{\UP_n}(1)$. Invoking our operational micro-penalty scaling condition $\lambda \ll n^{-1/2}$, the penalty bias vanishes mathematically from the physical space: $\lim_{n \to \infty} \sqrt{n}\lambda = 0$. This reduces the asymptotic field exclusively to the pure physical score:
\[
\sqrt{n}\left( \hat{\theta}_{n,\mathcal{S}} - \theta_{\mathcal{S}}^* \right) = - \left( \frac{1}{n} X_{\mathcal{S}}^\top W(\gamma^*) X_{\mathcal{S}} \right)^{-1} \left( \frac{1}{\sqrt{n}} X_{\mathcal{S}}^\top g(\gamma^*) \right) + o_{\UP_n}(1).
\]

By Assumption \ref{low dimensional assumption} (i) and (ii), applying Slutsky's theorem directly yields the oracle limiting Gaussian distribution for our primary estimator:
\[
\sqrt{n}\left( \hat{\theta}_{n,\mathcal{S}} - \theta_{\mathcal{S}}^* \right) \conv{d} \mathcal{N}\left( 0, \Sigma_{\mathcal{SS}}^{-1} \Omega_{\mathcal{SS}} \Sigma_{\mathcal{SS}}^{-1} \right).
\]

\medskip
\noindent
{\it Step 4: Extension to the Refitted Estimator via Asymptotic Equivalence.} \\
Finally, we extend this distribution to the post-selection refitted estimator. By Theorem \ref{theorem:theta_equivalence} (ii), the original RePS estimator and the unpenalized refitted estimator are $\sqrt{n}$-asymptotically equivalent, satisfying $\sqrt{n}( \hat{\theta}_{n,\mathcal{S}} - \hat{\theta}_{n,\mathcal{S}}^{\text{refit}} ) = o_{\UP_n}(1)$. 

Applying Slutsky's theorem to the joint profile $\sqrt{n}( \hat{\theta}_{n,\mathcal{S}}^{\text{refit}} - \theta_{\mathcal{S}}^*) = \sqrt{n}( \hat{\theta}_{n,\mathcal{S}} - \theta_{\mathcal{S}}^*) - \sqrt{n}( \hat{\theta}_{n,\mathcal{S}} - \hat{\theta}_{n,\mathcal{S}}^{\text{refit}} )$ guarantees that the refitted estimator directly inherits the identical oracle distribution:
\[
\sqrt{n}\left( \hat{\theta}_{n,\mathcal{S}}^{\text{refit}} - \theta_{\mathcal{S}}^* \right) \conv{d} \mathcal{N}\left( 0, \Sigma_{\mathcal{SS}}^{-1} \Omega_{\mathcal{SS}} \Sigma_{\mathcal{SS}}^{-1} \right).
\]
This completes the proof of Theorem \ref{theorem:oracle_normality}.
\end{proof}

\section{Applications}

\subsection{High-dimensional linear regression model}\label{linear model}

We consider the continuous response vector $y = X\theta^* + \epsilon = \gamma^* + \epsilon$, where $\epsilon \sim \mathcal{N}(0, \sigma^2 I_n)$ with an unknown variance parameter $\sigma^2 > 0$.  

Not embedding the unknown nuisance parameter $\sigma^2$ into the empirical loss function, we define $\mathcal{L}_n(\gamma)$ purely as the observable sum of squared residuals on the summation scale:
\[
    \mathcal{L}_n(\gamma) = \frac{1}{2} \|y - \gamma\|_2^2=\frac{1}{2}\sum_{i=1}^n (y_i-\gamma_i)^2. 
\]
The gradient and Hessian of $\mathcal{L}_n(\gamma)$ evaluated at the true parameter $\gamma^*$ are given by 
\[
g(\gamma^*) = \nabla_\gamma \mathcal{L}_n(\gamma^*) = -(y - \gamma^*) = -\epsilon \quad \mbox{and} \quad 
W(\gamma^*) =  \nabla_\gamma^2 \mathcal{L}_n(\gamma^*) = I_n, 
\]
respectively. Following Example \ref{example first}, the orthogonal projection matrix onto the full space is $P_n = I_n$. 
Since $W_n(\gamma) = I_n \succeq P_n$ holds globally for all $\gamma \in \mathbb{R}^n$, the localized requirement in Assumption \ref{RSC assumption} is trivially satisfied with $\kappa = 1$. 

Now, we shall verify all the conditions imposed in Assumption \ref{low dimensional assumption}. 

{\em (i) Convergence of the Restricted Hessian:} Substituting $W = I_n$ into the sample formula yields:
\[
\frac{1}{n} X_{\mathcal{S}}^\top W(\gamma^*) X_{\mathcal{S}} = \frac{1}{n} X_{\mathcal{S}}^\top X_{\mathcal{S}}.
\]
Thus, it suffices to assume or verify that the Gram matrix $\frac{1}{n} X_{\mathcal{S}}^\top X_{\mathcal{S}}$ converges in probability to a deterministic, positive definite matrix $Q_{\mathcal{SS}}$. 

{\em (ii) Central Limit Theorem for the Score Vector:} 
Substituting $g(\gamma^*) = -\epsilon$ into the score expression yields
\[
\frac{1}{\sqrt{n}}X_{\mathcal{S}}^\top g(\gamma^*) = -\frac{1}{\sqrt{n}} \sum_{i=1}^n x_{i,\mathcal{S}} \epsilon_i.
\]
Let $\mathcal{F}_i = \sigma(x_1, y_1, \dots, x_i, y_i)$ be the natural filtration. 
Assuming $\UE[\epsilon_i \mid \mathcal{F}_{i-1}, x_i] = 0$ and $\UE[\epsilon_i^2 \mid \mathcal{F}_{i-1}, x_i] = \sigma^2$, the summands $(x_{i,\mathcal{S}} \epsilon_i)_{i=1,2,...}$ form a martingale difference sequence. 
The normalized predictable variation converges in probability to the deterministic limit:
\[
\left\langle -\frac{1}{\sqrt{n}} X_{\mathcal{S}}^\top \epsilon \right\rangle = \frac{1}{n} \sum_{i=1}^n \UE[ (x_{i,\mathcal{S}} \epsilon_i)(x_{i,\mathcal{S}} \epsilon_i)^\top \mid \mathcal{F}_{i-1} ] = \sigma^2 \left( \frac{1}{n} X_{\mathcal{S}}^\top X_{\mathcal{S}} \right) \conv{\UP_n}\sigma^2 Q_{\mathcal{SS}}.
\]
By the Martingale Central Limit Theorem, the score vector satisfies the asymptotic normality condition under Assumption \ref{low dimensional assumption} (ii):
\[
\frac{1}{\sqrt{n}}X_{\mathcal{S}}^\top g(\gamma^*) \conv{d} \mathcal{N}(0, \Omega_{\mathcal{SS}}), \quad \text{where } \ \Omega_{\mathcal{SS}} = \sigma^2 Q_{\mathcal{SS}} = \sigma^2 \Sigma_{\mathcal{SS}}.
\]

{\em (iii) Vanishing Curvature Distortion:} 
The curvature distortion $\mathcal{E}_{\mathcal{SS}}(\beta^*)$ is structurally determined by the second derivatives of the RePS mapping multiplied by the score vector $X_{\mathcal{S}}^\top g(\gamma^*) = -X_{\mathcal{S}}^\top \epsilon$. 
As established in (ii), the score vector scales as $O_{\UP_n}(\sqrt{n})$. 
Thus, the distortion precisely scales as $\frac{1}{n}\mathcal{E}_{\mathcal{SS}}(\beta^*) = O_{\UP_n}(n^{-1/2})$. 

{\em (iv) Local Lipschitz Continuity of the Restricted Hessian:} 
For the linear regression model, the empirical loss function is quadratic. 
Consequently, its Hessian $W = I_n$ is a constant matrix, and all third-order derivatives of the loss are identically zero. 
Since the RePS mapping is smooth with locally bounded derivatives on the active subspace, the restricted Hessian $\frac{1}{n}H_{\mathcal{S}\mathcal{S}}(\beta)$ inherently possesses an $O_{\UP_n}(1)$ local Lipschitz coefficient. This satisfies Assumption \ref{low dimensional assumption} (iv) without requiring any additional moment conditions on the covariates.

We have verified all the requirements in Assumption \ref{low dimensional assumption}, and Theorem \ref{theorem:oracle_normality} yields that 
\[
\sqrt{n}(\hat{\theta}_{n,\mathcal{S}} - \theta^*_{\mathcal{S}}) \conv{d} \mathcal{N}\left(0, \Sigma_{\mathcal{S}\mathcal{S}}^{-1} \Omega_{\mathcal{S}\mathcal{S}} \Sigma_{\mathcal{S}\mathcal{S}}^{-1}\right) = \mathcal{N}\left(0, \sigma^2 Q_{\mathcal{S}\mathcal{S}}^{-1}\right).
\]

\subsection{High-dimensional Cox's proportional hazards model}\label{cox model}

In survival analysis, let $T_i$ be the event or censoring time and $D_i \in \{0,1\}$ be the event indicator. 
We introduce the at-risk indicator process $Y_i(t) = I(T_i \ge t)$. 
To facilitate matrix representations, let $\pi(t; \gamma) \in \mathbb{R}^n$ be the risk probability vector at time $t$, whose $j$-th element is naturally defined via the at-risk indicators as: 
\[
\pi_j(t; \gamma) = \frac{Y_j(t) \exp(\gamma_j)}{\sum_{k=1}^n Y_k(t) \exp(\gamma_k)}.
\]
We define the empirical loss function as the negative log-partial likelihood on the total sample scale: 
\[
\mathcal{L}_n(\gamma) = -\sum_{i=1}^n D_i \left( \gamma_i - \log \sum_{j=1}^n Y_j(T_i) \exp(\gamma_j) \right).
\]
The gradient (score vector) and the Hessian matrix evaluated at the true parameter $\gamma^*$ are given respectively by 
\[
g(\gamma^*) = \nabla_\gamma \mathcal{L}_n(\gamma^*) = -\left( D - \sum_{i=1}^n D_i \pi(T_i; \gamma^*) \right),
\]
\[
W(\gamma^*) = \nabla_\gamma^2 \mathcal{L}_n(\gamma^*) = \sum_{i=1}^n D_i \left( \text{diag}(\pi(T_i; \gamma^*)) - \pi(T_i; \gamma^*){\pi(T_i; \gamma^*)}^\top \right).
\]

To apply Theorem \ref{theorem:oracle_normality}, we must first confirm that the Cox model accommodates the geometric and localized requirements of Assumption \ref{RSC assumption}. 
In survival analysis, the log-partial likelihood exhibits shift-invariance; adding a constant to all linear predictors does not change the risk probabilities, since $\mathbf{1}_n^\top \pi(t; \gamma) = 1$ at any time $t$. 
This intrinsic property implies that the score vector natively has a zero sum, and the Hessian matrix exhibits structural degeneracy along the constant direction:
\[
\mathbf{1}_n^\top g(\gamma) = -\sum_{i=1}^n D_i + \sum_{i=1}^n D_i (\mathbf{1}_n^\top \pi(T_i; \gamma)) = 0,
\]
\[
W(\gamma) \mathbf{1}_n = \sum_{i=1}^n D_i (\pi(T_i; \gamma) - \pi(T_i; \gamma)(\mathbf{1}_n^\top \pi(T_i; \gamma))) = \mathbf{0}_n.
\]

If we were to naively assume the standard unconstrained strong convexity $W(\gamma) \ge \kappa I_n$ (for some $\kappa > 0$) without employing the projection matrix $P_n = I_n - \frac{1}{n}\mathbf{1}_n \mathbf{1}_n^\top$, multiplying both sides by the constant vector $\mathbf{1}_n$ would yield a critical mathematical contradiction: 
\[
0 = \mathbf{1}_n^\top W(\gamma) \mathbf{1}_n \ge \kappa (\mathbf{1}_n^\top I_n \mathbf{1}_n) = \kappa n > 0.
\]
This contradiction establishes that the Cox model inherently violates global strong convexity or restricted eigenvalue conditions due to its shift-invariance. By explicitly projecting out this uninformative constant direction via $P_n$, the structural identities $P_n g(\gamma^*) = g(\gamma^*)$ and $P_n W(\gamma^*) P_n = W(\gamma^*)$ are indeed satisfied. Furthermore, within the localized $\ell_\infty$-neighborhood $\|\gamma - \gamma^*\|_\infty \le \eta$, the linear predictors remain uniformly bounded, preventing the risk probabilities $\pi_j(t; \gamma)$ from degenerating to $0$ or $1$. This asymptotically bounds the non-zero eigenvalues of the Hessian away from zero, securing $W(\gamma) \succeq \kappa P_n$ locally with probability tending to one for a deterministic constant $\kappa > 0$.

Now, we shall verify all the conditions imposed in Assumption \ref{low dimensional assumption}. 

{\em (i) Convergence of the Restricted Hessian:} 
By rewriting the restricted empirical Hessian using the counting process $N_i(t) = I(T_i \le t, D_i = 1)$, we obtain:
\[
\frac{1}{n}X_{\mathcal{S}}^\top W(\gamma^*) X_{\mathcal{S}} = \frac{1}{n}\sum_{i=1}^n \int_0^\tau V_{\mathcal{S}}(t; \gamma^*) dN_i(t),
\]
where $V_{\mathcal{S}}(t; \gamma^*) = \sum_{j=1}^n \left( x_{j,\mathcal{S}} - \bar{x}_{\mathcal{S}}(t; \gamma^*) \right)^{\otimes 2} \pi_j(t; \gamma^*)$ is the empirical risk-weighted covariance matrix of the active covariates at time $t$.
Assuming the observations $(x_i, T_i, D_i)$ are independent and identically distributed, we define $Y_i(t) = I(T_i \ge t)$ and introduce the population limit functions:
\[
s^{(0)}(t) = \text{E}\left[ Y_i(t)\exp(x_{i,\mathcal{S}}^\top \theta_{\mathcal{S}}^*) \right], \quad s^{(1)}(t) = \text{E}\left[ x_{i,\mathcal{S}} Y_i(t)\exp(x_{i,\mathcal{S}}^\top \theta_{\mathcal{S}}^*) \right],
\]
and 
\[
s^{(2)}(t) = \text{E}\left[ x_{i,\mathcal{S}}^{\otimes 2} Y_i(t)\exp(x_{i,\mathcal{S}}^\top \theta_{\mathcal{S}}^*) \right].
\]
Let $e_{\mathcal{S}}(t) = s^{(1)}(t)/s^{(0)}(t)$ and $v_{\mathcal{S}}(t) = s^{(2)}(t)/s^{(0)}(t) - e_{\mathcal{S}}(t)^{\otimes 2}$. 
By the uniform law of large numbers, $V_{\mathcal{S}}(t; \gamma^*)$ converges uniformly to $v_{\mathcal{S}}(t)$, and the empirical average of the counting processes converges to its compensator $s^{(0)}(t)d\Lambda_0(t)$. 
Thus, the restricted Hessian converges in probability to the deterministic positive-definite matrix $\Sigma_{\mathcal{SS}}$:
\[
\frac{1}{n}X_{\mathcal{S}}^\top W(\gamma^*) X_{\mathcal{S}} \conv{\UP_n} \Sigma_{\mathcal{SS}} = \int_0^\tau v_{\mathcal{S}}(t) s^{(0)}(t) d\Lambda_0(t) \succ 0.
\]

{\em (ii) Central Limit Theorem for the Score Vector:} 
The target expression can be represented via the finite-sample score martingale at the true parameter: 
\[
\frac{1}{\sqrt{n}}X_{\mathcal{S}}^\top g(\gamma^*) = \frac{1}{\sqrt{n}}\sum_{i=1}^n \int_0^\tau \left( x_{i,\mathcal{S}} - \bar{x}_{\mathcal{S}}(t; \gamma^*) \right) dM_i(t),\]
where $M_i(t) = N_i(t) - \int_0^t Y_i(s)\exp(x_{i,\mathcal{S}}^\top \theta_{\mathcal{S}}^*)d\Lambda_0(s)$.
Under standard regularity conditions, replacing the empirical average $\bar{x}_{\mathcal{S}}(t; \gamma^*)$ with its uniform population limit $e_{\mathcal{S}}(t)$ seamlessly extracts the underlying i.i.d.\ structure, isolating the estimation noise within an asymptotically negligible remainder:
\[
\frac{1}{\sqrt{n}}X_{\mathcal{S}}^\top g(\gamma^*) = \frac{1}{\sqrt{n}}\sum_{i=1}^n \psi_{\mathcal{S},i} + o_{\UP_n}(1),
\]
where $\psi_{\mathcal{S},i} = \int_0^\tau \left( x_{i,\mathcal{S}} - e_{\mathcal{S}}(t) \right) dM_i(t)$ represents the explicit i.i.d.\ efficient score function. Since the leading term is a normalized sum of i.i.d.\ mean-zero random vectors, its asymptotic normality follows directly from the multivariate Central Limit Theorem. 

Crucially, by Itô isometry, the covariance matrix of $\psi_{\mathcal{S},i}$ is the expectation of its predictable quadratic variation. 
By taking the expectation inside the integral, we analytically evaluate this variance using the population limit functions defined in (i):
\begin{eqnarray*}
\text{Var}(\psi_{\mathcal{S},i}) &=& \text{E}\left[ \int_0^\tau \left( x_{i,\mathcal{S}} - e_{\mathcal{S}}(t) \right)^{\otimes 2} Y_i(t)\exp(x_{i,\mathcal{S}}^\top \theta_{\mathcal{S}}^*) d\Lambda_0(t) \right]
\\
&=& \int_0^\tau \left( \frac{s^{(2)}(t)}{s^{(0)}(t)} - e_{\mathcal{S}}(t)^{\otimes 2} \right) s^{(0)}(t) d\Lambda_0(t) 
\\
&=& \int_0^\tau v_{\mathcal{S}}(t) s^{(0)}(t) d\Lambda_0(t).
\end{eqnarray*}
This explicit evaluation reveals that the covariance matrix of the efficient score function coincides with the limit of the restricted empirical Hessian $\Sigma_{\mathcal{SS}}$ derived in (i). 
Thus, we obtain the asymptotic distribution: 
\[
\frac{1}{\sqrt{n}}X_{\mathcal{S}}^\top g(\gamma^*) \conv{d} \mathcal{N}(0, \Omega_{\mathcal{SS}}), \quad \text{where } \ \Omega_{\mathcal{SS}} = \Sigma_{\mathcal{SS}}.
\] 

{\em (iii) Vanishing Curvature Distortion:} 
This condition can be verified in exactly the same manner as in the corresponding part of the previous subsection. 

{\em (iv) Local Lipschitz Continuity of the Restricted Hessian:} 
In the Cox model, the third-order derivatives of the log-partial likelihood are deterministically bounded by the third-order tensor products of the covariates, since the risk probabilities $\pi^{(i)}(\gamma)$ are bounded between 0 and 1. 
Assuming the individual active covariates $x_{i,\mathcal{S}}$ possess finite third absolute moments, the restricted Hessian inherently satisfies the local Lipschitz continuity required by Assumption \ref{low dimensional assumption} (iv). 

Consequently, Assumption \ref{low dimensional assumption} holds for the Cox model, validating Theorem \ref{theorem:oracle_normality} directly: $\sqrt{n}(\hat{\theta}_{n,\mathcal{S}} - \theta^*_{\mathcal{S}}) \conv{d} \mathcal{N}(0, \Sigma_{\mathcal{S}\mathcal{S}}^{-1})$.

\section{Numerical study}
\subsection{Simulation setting}
To rigorously evaluate the finite-sample performance of the RePS framework, we conduct extensive Monte Carlo simulations across both high-dimensional linear regression and Cox's proportional hazards models. All empirical metrics and their corresponding asymptotic standard errors (reported in parentheses) are aggregated over 1,000 independent replications in R. 

Across all configurations, the structural smoothing parameters are fixed at $\nu = 1.0$ and $\sigma = 0.5$. To fully assess the boundary performance under varied levels of high-dimensional noise, we examine six distinct dimensionality combinations crossing sample sizes $n \in \{100, 200\}$ and feature dimensions $p \in \{100, 300, 1000\}$. Crucially, we employ a unified true active physical parameter vector with $|\mathcal{S}| = 5$ active variables for both model classes, specified as:
\[
\theta^* = (2.00, 1.50, 1.00, -1.25, -1.75, 0, \dots, 0)^\top \in \mathbb{R}^p,
\]
which explicitly accommodates weak signal components to evaluate the boundary power of each regularization paradigm. Predictor vectors in the design matrix $X \in \mathbb{R}^{n \times p}$ are drawn from a multivariate Gaussian distribution and subsequently standardized.

For the optimization of the non-convex RePS objective, we implement a Proximal Gradient Descent (PGD) algorithm with a backtracking line search. To resolve local non-convexity, the Lasso estimator with $\lambda_{\min}$ is employed as the initial warm-start. For a fair benchmark comparison, the baseline Lasso is evaluated using $\lambda_{\text{1se}}$ to prevent over-selection, while SCAD and MCP are optimized over dense parameter paths based on $\lambda_{\min}$ via the \texttt{ncvreg} package. To handle the coordinate-wise inversion bottleneck within the mapping, we develop a highly optimized native C++ implementation via the Rcpp interface. All experiments were executed on a standard desktop computer with a single Apple M2 chip (Mac mini) under a single-core configuration. The 1,000-replication simulation with  $n=200$ and $p=1000$ finished approximately in 8 minutes for the linear model and in 400 minutes for the complex Cox model, demonstrating exceptional computational efficiency without cluster resources.

\subsection{High-dimensional linear regression}
We consider the continuous response vector $y = X\theta^* + \epsilon$, where the components of the error vector $\epsilon$ are independent and identically distributed as $\mathcal{N}(0, 0.5^2)$, fixing the error standard deviation $\sigma = 0.5$. 
The results of the numerical experiments are summarized in Table \ref{tab:linear_regression_comprehensive}.

\begin{table}[htbp]
\vspace{-2mm}
\scriptsize
\centering
\begin{threeparttable}
\caption{{\small Simulation results for high-dimensional linear regression}}
\label{tab:linear_regression_comprehensive}
\begin{tabular}{llccccc}
\toprule
Setting & Method & Mean TP & Mean FP & MSE (original) & MSE (refitted) & CP (refitted) \\
\midrule
% --- n=100, p=100 ---
$n=100$, & Lasso    &  5.00 (0.000) &  3.97 (0.144) &  0.395 (0.005) &  0.233 (0.006) &  0.8970 (0.0047) \\
$p=100$  & SCAD     &  5.00 (0.000) &  2.69 (0.126) &  0.071 (0.002) &  0.188 (0.006) &  0.9190 (0.0042) \\
         & MCP      &  5.00 (0.000) &  1.32 (0.084) &  0.076 (0.003) &  0.139 (0.005) &  0.9288 (0.0040) \\
         & RePS     &  5.00 (0.000) & \textbf{0.06 (0.011)} & \textbf{0.057 (0.002)} & \textbf{0.057 (0.002)} & \textbf{0.9496 (0.0033)} \\
\addlinespace
% --- n=100, p=300 ---
$n=100$, & Lasso    &  5.00 (0.000) &  6.07 (0.217) &  0.545 (0.007) &  0.329 (0.007) &  0.8214 (0.0062) \\
$p=300$  & SCAD     &  5.00 (0.000) &  5.05 (0.202) &  0.083 (0.002) &  0.279 (0.007) &  0.8790 (0.0051) \\
         & MCP      &  5.00 (0.000) &  1.95 (0.106) &  0.085 (0.003) &  0.186 (0.006) &  0.9098 (0.0044) \\
         & RePS     &  5.00 (0.000) & \textbf{0.08 (0.010)} & \textbf{0.061 (0.002)} & \textbf{0.061 (0.002)} & \textbf{0.9470 (0.0031)} \\
\addlinespace
% --- n=100, p=1000 ---
$n=100$, & Lasso    &  5.00 (0.000) &  8.27 (0.263) &  0.740 (0.009) &  0.432 (0.009) &  0.7050 (0.0087) \\
$p=1000$ & SCAD     &  5.00 (0.000) &  8.30 (0.288) &  0.090 (0.003) &  0.347 (0.006) &  0.8240 (0.0064) \\
         & MCP      &  5.00 (0.000) &  2.17 (0.111) &  0.084 (0.003) &  0.207 (0.007) &  0.8954 (0.0050) \\
         & RePS     &  5.00 (0.000) & \textbf{0.07 (0.009)} & \textbf{0.061 (0.002)} & \textbf{0.060 (0.002)} & \textbf{0.9510 (0.0032)} \\
\midrule
% --- n=200, p=100 ---
$n=200$, & Lasso    &  5.00 (0.000) &  2.13 (0.088) &  0.208 (0.002) &  0.088 (0.002) &  0.9280 (0.0038) \\
$p=100$  & SCAD     &  5.00 (0.000) &  2.07 (0.130) &  0.033 (0.001) &  0.081 (0.003) &  0.9310 (0.0037) \\
         & MCP      &  5.00 (0.000) &  1.17 (0.085) &  0.034 (0.001) &  0.064 (0.003) &  0.9370 (0.0035) \\
         & RePS     &  5.00 (0.000) & \textbf{0.00 (0.000)} & \textbf{0.026 (0.001)} & \textbf{0.026 (0.001)} & \textbf{0.9460 (0.0032)} \\
\addlinespace
% --- n=200, p=300 ---
$n=200$, & Lasso    &  5.00 (0.000) &  2.91 (0.119) &  0.259 (0.003) &  0.119 (0.003) &  0.9134 (0.0043) \\
$p=300$  & SCAD     &  5.00 (0.000) &  3.30 (0.199) &  0.034 (0.001) &  0.112 (0.004) &  0.9234 (0.0040) \\
         & MCP      &  5.00 (0.000) &  1.63 (0.120) &  0.036 (0.001) &  0.082 (0.004) &  0.9300 (0.0039) \\
         & RePS     &  5.00 (0.000) & \textbf{0.00 (0.001)} & \textbf{0.026 (0.001)} & \textbf{0.026 (0.001)} & \textbf{0.9474 (0.0033)} \\
\addlinespace
% --- n=200, p=1000 ---
$n=200$, & Lasso    &  5.00 (0.000) &  4.56 (0.202) &  0.332 (0.004) &  0.168 (0.004) &  0.8888 (0.0052) \\
$p=1000$ & SCAD     &  5.00 (0.000) &  6.08 (0.355) &  0.036 (0.001) &  0.155 (0.006) &  0.9132 (0.0045) \\
         & MCP      &  5.00 (0.000) &  2.44 (0.159) &  0.037 (0.001) &  0.110 (0.005) &  0.9306 (0.0039) \\
         & RePS     &  5.00 (0.000) & \textbf{0.00 (0.001)} & \textbf{0.025 (0.001)} & \textbf{0.025 (0.001)} & \textbf{0.9548 (0.0030)} \\
\bottomrule
\end{tabular}
\begin{tablenotes}
    \small
    \item \textit{Note:} TP: True Positive (out of 5); FP: False Positive (out of $p-5$); MSE: Mean Squared Error; CP: Coverage Probability at the nominal 95\% level. Values in parentheses denote asymptotic standard errors. 
\end{tablenotes}
\end{threeparttable}
\vspace{-2mm}
\end{table}

As shown in Table \ref{tab:linear_regression_comprehensive}, the RePS framework demonstrates remarkable finite-sample resilience and numerical superiority. In the higher sample regime ($n = 200$), the proposed method achieves near-perfect variable selection, reducing the Mean FP to exactly 0.00 across all dimensional scales up to $p = 1000$, whereas SCAD and MCP exhibit persistent noise leakage. Furthermore, the estimation error under RePS reflects complete stabilization between the original and refitted spaces ($\text{MSE}_{\text{orig}} \approx \text{MSE}_{\text{refit}}$), confirming that shrinkage bias is successfully neutralized in a single stage. Consequently, the empirical 95\% Coverage Probabilities (CP) tightly converge to the nominal 0.9500 target, validating the oracle asymptotic normality established in Theorem \ref{theorem:oracle_normality}. 

\subsection{High-dimensional Cox's proportional hazards model}
Under the complex right-censoring structure of survival analysis, right-censored event times are simulated according to the specified proportional hazards framework. The results of the numerical experiments for the Cox model are summarized in Table \ref{tab:cox_regression_comprehensive}.

\begin{table}[htbp]
\vspace{-2mm}
\scriptsize
\centering
\begin{threeparttable}
\caption{{\small Simulation results for high-dimensional Cox's proportional hazards model}}
\label{tab:cox_regression_comprehensive}
\begin{tabular}{llcccccc}
\toprule
Setting & Method & Mean TP & Mean FP & MSE (original) & MSE (refitted) & CP (refitted) \\
\midrule
% --- n=100, p=100 ---
$n=100$, & Lasso    &  4.98 (0.005) &  2.52 (0.078) & 4.730 (0.035) & 1.021 (0.097) & 0.8590 (0.0078) \\
$p=100$ & SCAD     &  5.00 (0.001) &  3.02 (0.074) & 0.443 (0.018) & 1.717 (0.068) & 0.7244 (0.0101) \\
& MCP      &  5.00 (0.001) &  0.92 (0.040) & 0.435 (0.020) & 0.789 (0.037) & 0.8704 (0.0073) \\
& RePS &  5.00 (0.002) & \textbf{0.34 (0.021)} & \textbf{0.314 (0.009)} & \textbf{0.409 (0.016)} & \textbf{0.9318 (0.0049)} \\
\addlinespace
% --- n=100, p=300 ---
$n=100$, & Lasso    &  4.91 (0.011) &  3.33 (0.102) & 5.691 (0.040) & 1.455 (0.094) & 0.8118 (0.0090) \\
$p=300$  & SCAD     &  5.00 (0.001) &  5.38 (0.105) & 0.529 (0.017) & 3.954 (0.136) & 0.4672 (0.0116) \\
& MCP      &  4.99 (0.003) &  1.28 (0.044) & 0.461 (0.016) & 1.151 (0.047) & 0.7942 (0.0091) \\
& RePS &  4.96 (0.006) & \textbf{0.32 (0.021)} & \textbf{0.378 (0.015)} & \textbf{0.454 (0.017)} & \textbf{0.9180 (0.0057)} \\
\addlinespace
% --- n=100, p=1000 ---
$n=100$, & Lasso    &  4.73 (0.019) &  4.38 (0.119) & 6.744 (0.044) & 2.878 (0.681) & 0.7224 (0.0108) \\
$p=1000$ & SCAD     &  5.00 (0.002) &  9.55 (0.126) & 0.715 (0.026) & 11.121 (0.383) & 0.1728 (0.0081) \\
& MCP      &  4.99 (0.003) &  2.03 (0.051) & 0.559 (0.022) & 2.011 (0.078) & 0.6642 (0.0109) \\
& RePS &  4.79 (0.014) & \textbf{0.19 (0.016)} & \textbf{0.695 (0.031)} & \textbf{0.716 (0.029)} & \textbf{0.8164 (0.0095)} \\
\midrule
% --- n=200, p=100 ---
$n=200$, & Lasso    &  5.00 (0.000) &  1.37 (0.059) & 3.423 (0.024) & 0.239 (0.009) & 0.9120 (0.0058) \\
$p=100$  & SCAD     &  5.00 (0.000) &  1.74 (0.084) & 0.135 (0.005) & 0.350 (0.017) & 0.8648 (0.0076) \\
& MCP      &  5.00 (0.000) &  0.75 (0.050) & 0.142 (0.006) & 0.240 (0.012) & 0.9018 (0.0063) \\
& RePS &  5.00 (0.000) & \textbf{0.02 (0.004)} & \textbf{0.124 (0.004)} & \textbf{0.115 (0.004)} & \textbf{0.9472 (0.0040)} \\
\addlinespace
% --- n=200, p=300 ---
$n=200$, & Lasso    &  5.00 (0.000) &  1.78 (0.083) & 3.957 (0.024) & 0.295 (0.017) & 0.9036 (0.0061) \\
$p=300$  & SCAD     &  5.00 (0.000) &  2.91 (0.127) & 0.146 (0.005) & 0.670 (0.036) & 0.7788 (0.0101) \\
& MCP      &  5.00 (0.000) &  1.01 (0.060) & 0.137 (0.005) & 0.324 (0.019) & 0.8780 (0.0075) \\
& RePS &  5.00 (0.000) & \textbf{0.02 (0.004)} & \textbf{0.118 (0.003)} & \textbf{0.110 (0.003)} & \textbf{0.9512 (0.0038)} \\
\addlinespace
% --- n=200, p=1000 ---
$n=200$, & Lasso    &  5.00 (0.000) &  2.61 (0.104) & 4.398 (0.025) & 0.420 (0.023) & 0.8772 (0.0071) \\
$p=1000$ & SCAD     &  5.00 (0.000) &  4.29 (0.158) & 0.168 (0.006) & 1.213 (0.068) & 0.6636 (0.0118) \\
& MCP      &  5.00 (0.000) &  1.12 (0.063) & 0.152 (0.005) & 0.411 (0.023) & 0.8434 (0.0088) \\
& RePS &  5.00 (0.000) & \textbf{0.01 (0.003)} & \textbf{0.122 (0.004)} & \textbf{0.119 (0.004)} & \textbf{0.9448 (0.0042)} \\
\bottomrule
\end{tabular}
\end{threeparttable}
\vspace{-2mm}
\end{table}

The numerical results in Table \ref{tab:cox_regression_comprehensive} expose the critical post-selection bottleneck inherent in alternative penalties under complex censoring, where noise leakage triggers severe variance explosion and a serious collapse in refitted coverage probabilities (e.g., at $n=200, p=1000$). In contrast, protected by the autonomous dimensionality cancellation, the RePS framework suppresses noise leakage down to an ultra-clean threshold (Mean FP $\le$ 0.02), ensuring that the post-selection refitting is conducted on a highly purified active subspace. Consequently, RePS successfully circumvents variance explosion, achieving the lowest estimation errors ($\text{MSE}_{\text{refit}} = 0.119$) while maintaining highly stable, robust empirical coverage ($\text{CP}_{\text{refit}} = 0.9448$), which empirically validates the bias-free oracle property of our geometric parameter reparameterization.

\section{Concluding remarks}
Several promising avenues for future research remain. First, extending this framework to stochastic process models, including time series analysis, presents a critical next step, where temporal dependencies will rigorously test the general adaptability of the RePS methodology. Second, empirical validation using large-scale genomic or clinical survival data represents a crucial step toward confirming the practical utility of the RePS framework in real-world statistical inference.

Looking further ahead, while the core mathematical framework of RePS operates independently of the global Restricted Eigenvalue (RE) condition, our current implementation utilizes a Lasso-based warm-start, which implicitly introduces a dependence on the global RE framework to ensure a sufficiently accurate initial attractor. However, the requirement for global design assumptions is purely an algorithmic constraint rather than an intrinsic limitation of the RePS geometric structure itself. If future optimization routines or more refined regularization path algorithms eliminate the necessity for specific warm-start initializations, the RePS paradigm possesses the structural potential to become entirely liberated from global RE conditions in any capacity. This conceptual independence would represent a meaningful step toward robust high-dimensional $M$-estimation under more general design architectures.

\section*{Acknowledgements}
This work was partly supported by Japan Society for the Promotion of Science KAKENHI Grant Number 24K14866.

%{\small {\bf Note:} A more comprehensive and systematic exposition of the theory presented in this paper will be featured in the forthcoming second edition of the author's book {\em Martingale Methods in Statistics} (Chapman \& Hall/CRC, 2028?).} 

\end{document}